\documentclass[a4paper,11pt]{article}

\usepackage[a4paper, margin={1in, 1in}]{geometry}
\usepackage{mathtools}
\usepackage{algorithmic}
\usepackage{algorithm}
\usepackage{amsthm}
\usepackage{amssymb}
\usepackage{amsmath}
\usepackage{enumitem}
\usepackage{float}
\usepackage{etoolbox}% http://ctan.org/pkg/etoolbox
\usepackage{color}
\patchcmd{\thmhead}{(#3)}{#3}{}{}

\theoremstyle{plain}
\newtheorem{thm}{Theorem}[section]

\newtheorem{lem}[thm]{Lemma}
\newtheorem{cor}[thm]{Corollary}
\newtheorem{obs}[thm]{Observation}

\theoremstyle{definition}
\newtheorem{definition}[thm]{Definition}

\theoremstyle{remark}

\newtheorem*{remark*}{Remark}

\newcommand{\Par}{\text{Par}}
\newcommand{\Pre}{\text{Pre}}
\newcommand{\acc}{\text{acc}}
\newcommand{\rej}{\text{rej}}

\begin{document}

\title{Testing local properties of arrays}
\author{Omri Ben-Eliezer\thanks{Blavatnik School of
		Computer Science, Tel Aviv University, Tel Aviv 69978, Israel.
		{\tt omrib@mail.tau.ac.il}.}}
\date{}
\maketitle
\begin{abstract}
We study testing of local properties in one-dimensional and multi-dimensional arrays. A property of $d$-dimensional arrays $f:[n]^d \to \Sigma$ is \emph{$k$-local} if it can be defined by a family of $k \times \ldots \times k$ \emph{forbidden consecutive patterns}. This definition captures numerous interesting properties. For example, monotonicity, Lipschitz continuity and submodularity are $2$-local; convexity is (usually) $3$-local; and many typical problems in computational biology and computer vision involve $o(n)$-local properties.

In this work, we present a generic approach to test all local properties of arrays over any finite (and not necessarily bounded size) alphabet. We show that \emph{any} $k$-local property of $d$-dimensional arrays is testable by a simple canonical one-sided error non-adaptive $\varepsilon$-test, whose query complexity is $O(\epsilon^{-1}k \log{\frac{\epsilon n}{k}})$ for $d = 1$ and $O(c_d \epsilon^{-1/d} k \cdot n^{d-1})$ for $d > 1$. The queries made by the canonical test constitute sphere-like structures of varying sizes, and are completely independent of the property and the alphabet $\Sigma$. The query complexity is optimal for a wide range of parameters: For $d=1$, this matches the query complexity of many previously investigated local properties, while for $d > 1$ we design and analyze new constructions of $k$-local properties whose one-sided non-adaptive query complexity matches our upper bounds. For some previously studied properties, our method provides the first known sublinear upper bound on the query complexity.
\end{abstract}

\section{Introduction}
Property testing \cite{ GoldreichGoldwasserRon1998, RubinfeldSudan1996} is devoted to understanding how much information one needs to extract from an object in order to determine whether it satisfies a given property or is \emph{far} from satisfying the property. This active research area has seen many developments through the last two decades; see the recent book of Goldreich \cite{Goldreich2017} for a good introduction. The property testing notation we use here is standard, see Subsection \ref{subsec:notation} for the relevant definitions.

In this paper we focus on testing of local properties in structured data. The objects we consider are \emph{$d$-dimensional arrays}, where $d$ is a positive integer, viewed as a constant. 
A $d$-dimensional \emph{array} of \emph{width} $n$, or an \emph{$[n]^d$-array} in short, is a function $A \colon [n]^d \to \Sigma$ from the hypergrid $[n]^d$ to the \emph{alphabet}  $\Sigma$, where the alphabet $\Sigma$ is allowed to be any (arbitrarily large) finite set; we stress that the size of $\Sigma$ is usually \emph{not required} to be bounded as a function of the other parameters. For example, a \emph{string} is an $[n]^1$-array, and the commonly used RGB representation of images is basically an $[n]^2$-array over $\{0,1,\ldots, 255\}^3$, where the three values corresponding to each pixel represent the intensity of red, green and blue in it.

\subsection{Local properties}
\label{subsec:definition_local}
Korman, Reichman, and the author \cite{BenEliezerKormanReichman2017} recently investigated testing of the property of \emph{consecutive pattern freeness}, i.e., not containing a copy of some (predefined) ``forbidden'' consecutive subarray. Here, a $[k]^d$-array $S$ is a (consecutive) \emph{subarray} of an $[n]^d$-array $A$ in \emph{location}  $(i_1, \ldots, i_d) \in [n-k+1]^d$ if $A(i_1 + j_1 - 1, \ldots, i_d + j_d - 1) = S(j_1, \ldots, j_d)$ for any $j_1, \ldots, j_d \in [k]$.

Naturally, a more general follow-up question raised in \cite{BenEliezerKormanReichman2017} was the following: what can be said about testing of properties defined by a \emph{family of forbidden consecutive patterns}? As we shall see soon, many interesting properties of arrays (including a large fraction of the array properties that were previously investigated in the literature) can be characterized this way.     

With this in mind, we call a property \emph{local} if it can be characterized by a family of small forbidden consecutive patterns. Formally, a property $\mathcal{P}$ of $[n]^d$-arrays over an alphabet $\Sigma$ is \emph{$k$-local} (for $2 \leq k \leq n$) if there exists a family $\mathcal{F}$ of $[k]^d$-arrays over $\Sigma$ so that the following holds for any $[n]^d$-array $A$ over $\Sigma$:
\begin{center}
$A$ satisfies $\mathcal{P}$ $\iff$ None of the (consecutive) subarrays of $A$ is in $\mathcal{F}$.
\end{center}
For $\mathcal{P}$ as above, we sometimes write $\mathcal{P} = \mathcal{P}(\mathcal{F})$ to denote that $\mathcal{P}$ is defined by the forbidden family $\mathcal{F}$.

The main contribution of this work is a generic one-sided error non-adaptive framework to test $k$-local properties. In some cases, our method either matches or beats the best known upper bounds on the query complexity (although the running time might be far from optimal in general). We show the optimality of our method by proving a matching lower bound for non-adaptive one-sided tests, as well as a (weaker) lower bound for two-sided tests.    

In order to demonstrate the wide range of properties captured by the above definition, we now present various examples of properties that are $k$-local for small $k$, including some of the most widely investigated properties in the property testing literature, as well as properties from areas of computer science that were not systematically studied in the context of property testing. 
In what follows, the sum of two tuples $x = (x_1, \ldots, x_d), y = (y_1, \ldots, y_d)$ is defined as the tuple $(x_1 + y_1, \ldots, x_d + y_d)$; additionally, $e^i$ denotes the $i$-th unit vector in $d$ dimensions.

\begin{description}
	\item[Monotonicity] Perhaps the most thoroughly investigated property in the testing literature: see e.g. the  entries related to monotonicity testing in the Encyclopedia of Algorithms \cite{Chakrabarty2016, Raskhodnikova2016} and the references within. An $[n]^d$-array $A$ over an ordered alphabet $\Sigma$ is \emph{monotone} (non-decreasing) if $A(x) \leq A(y)$ for any $x = (x_1, \ldots, x_d)$ and $y = (y_1, \ldots, y_d)$ satisfying $x_i \leq y_i$ for any $i$. 
	Monotonicity is 2-local: an array $A$ is monotone if and only if there is no pair $x, x+e^i \in [n]^d$ so that $A(x) > A(x + e^i)$.
	
	\item[Lipschitz continuity] Another well-investigated property with connections to differential privacy  \cite{AwasthiJMR2016, BermanRaskhodnikovaYaroslavtev2014, ChakrabartyDJS2017, JhaRaskhodnikova2011}, an $[n]^d$-array $A$ is \emph{$c$-Lipschitz continuous} if $|A(x) - A(y)| \leq c \sum_{i=1}^{d} |y_i - x_i|$ for any $x, y \in [n]^d$. This condition holds iff $|A(x) - A(x+e^i)| \leq c$ for any $x, x+e^i \in [n]^d$, and thus Lipschitz continuity is also $2$-local.
	
	\item[Convexity] Discrete convexity is an important geometric property with connections to optimization and other areas \cite{BMR2015, BermanMurzabulatovRaskhodnikova2016, BlaisRY2014, ChenFreilichServedioSun2017, ParnasRR04, RademacherVempala2004}. A one-dimensional array $A$ is \emph{convex} if $\lambda A(x) + (1-\lambda)A(y) \geq A(\lambda x + (1- \lambda) y)$ for any $x, y \in [n]$ and $0 < \lambda < 1$ satisfying $\lambda x + (1- \lambda) y \in [n]$. 
	Convexity is 3-local for the case $d = 1$: an array $A:[n] \to \Sigma$ is convex if and only $A[x] - 2A[x+1] + A[x+2] \geq 0$ for any $x \in [n-2]$. 
	In higher dimensions, several different notions of discrete convexity have been used in the literature -- see e.g. the introductory sections of the book of Murota on discrete convex analysis \cite{Murota2003}. Two of the commonly used definitions, $M^{\sharp}$-convexity and $L^{\sharp}$-convexity, are $3$-local and $4$-local, respectively: see Theorems 4.1 and 4.2 in \cite{MoriguchiMurota2011}, where it is shown that both notions can be defined locally using slight variants of the Hessian matrix consisting of the partial discrete derivatives. Another common definition that is a natural variant of the continuous case states that convexity is equivalent to the positive semi-definiteness of the Hessian matrix; under this definition, convexity is $3$-local. 
	A strictly weaker notion of convexity, called \emph{separate convexity} \cite{BlaisRY2014}, is defined as follows: an $[n]^d$-array $A$ is separately convex if it is convex along each of the axes. Similarly to one-dimensional convexity, separate convexity is $3$-local for any $d$.
	
	\item[Properties of higher order derivatives] More generally, \emph{any} property of arrays that can be characterized by ``forbidden pointwise behavior'' of the first $k$ discrete derivatives \cite{BlaisRY2014} is $(k+1)$-local. Monotonicity (for $k=1$), Lipschitz continuity ($k=1$) and convexity ($k=2$) are special cases of such properties.
	
	\item[Submodularity] Another important property closely related to convexity \cite{BermanRaskhodnikovaYaroslavtev2014, BlaisBommireddi2017, ParnasRR04, SeshadhriVondrak2014}. Given $x = (x_1, \ldots, x_d), y = (y_1, \ldots, y_d) \in [n]^d$, define $x \land y = (\min(x_1, y_1), \ldots, \min(x_d, y_d))$ and $x \lor y = (\max(x_1, y_1), \ldots, \max(x_d, y_d))$. An $[n]^d$-array is \emph{submodular} if $A[x \land y] + A[x \lor y] \leq A[x] + A[y]$ for any $x, y \in [n]^d$. Submodularity is $2$-local: it is not hard to verify that submodularity is equivalent to the condition that $A(x) + A(x + e^i + e^j) \leq A(x + e^i) + A(x + e^j)$ for all $x$.
	
	\item[Pattern matching and computer vision]
	Tasks involving pattern matching under some limitations -- such as noise in the image, obstructed view, or rotation of elements in the image -- are at the core of computer vision and its applications. For example, the local property of not containing a good enough $\ell_1$-approximation of a given forbidden pattern is of practical importance in computer vision. Sublinear approaches closely related to property testing are known to be effective for problems of this type, see e.g. \cite{KormanRTA2017}.
	
	\item[Computational biology] Many problems in computational biology are closely related to one-dimensional pattern matching. As an example, a defensive mechanism of the human body against RNA-based viruses involves ``cutting'' a suspicious RNA fragment, if it finds one of a (small) family of short forbidden consecutive patterns in it, indicating that this RNA might belong to a virus. Thus, in order to generate fragments of RNA that are not destroyed by such defensive mechanisms (which is a basic task in computational biology),  understanding the process of ``repairing'' a fragment so that it will not contain any of the forbidden patterns is an interesting problem related to property testing. 
\end{description}

\subsection{Previous results on local properties}
\label{subsec:previous_results_local}
\paragraph{One-dimensional arrays}
A seminal result of Erg\"un et al. \cite{ErgunKannanKRV2000} shows that for constant $\epsilon$, monotonicity is $\epsilon$-testable over the line (that is, for one-dimensional arrays) using $O(\log n)$ queries over general alphabets. The non-adaptive one-sided error test proposed in \cite{ErgunKannanKRV2000} is based, roughly speaking, on imitating a binary search non-adaptively. It was shown by Fischer \cite{Fischer2004} that the above is tight even for two-sided error adaptive tests, proving a matching $\Omega(\log n)$ lower bound. Later on, Parnas, Ron and Rubinfeld \cite{ParnasRR04} and Jha and Raskhodnikova \cite{JhaRaskhodnikova2011} showed that the $O(\log n)$ upper bound on the non-adaptive one-sided query complexity also holds for convexity and Lipschitz continuity, respectively. For general $\epsilon$, the upper bound in \cite{JhaRaskhodnikova2011} is of the type $O(\epsilon^{-1} \log n)$; the same work also presents a matching lower bound of $\Omega(\log n)$ for the one-sided non-adaptive case, while $\Omega(\log n)$ lower bounds for two-sided non-adaptive tests of convexity, and more generally, monotonicity of the $\ell$-th derivative, are proved by Blais, Raskhodnikova and Yaroslavtsev  \cite{BlaisRY2014} using a communication complexity based approach \cite{BlaisBM2012}. Finally, a recent result of Belovs \cite{Belovs2018} refines the one-sided non-adaptive query complexity of monotonicity to $O(\epsilon^{-1}\log{\epsilon n})$. 

When the alphabet is binary (of size two), general positive results are known regarding the testability of local properties in one dimension. It follows from the testability of regular languages, established by Alon et al.\@ \cite{Alon2001}, that any $k$-local property is testable in $O(c(\mathcal{F}) \epsilon^{-1} (\log^3 (\epsilon^{-1})))$ queries, where $c(\mathcal{F})$ depends only on the family $\mathcal{F}$ of forbidden consecutive length-$k$ patterns defining the property. However, $c(\mathcal{F})$ can be exponential in $k$ in general. 

\paragraph{Multi-dimensional arrays}
Chakrabarty and Seshadhri \cite{ChakrabartySeshadhri2013} extended some of the above results to hypergrids, showing that a general class of so-called ``bounded derivative'' properties (all of which are $2$-local), including monotonicity and Lipschitz continuity as special cases, are all testable over $[n]^d$-arrays with $O(\epsilon^{-1} d \log n)$ queries. Another work by the same authors \cite{ChakrabartySeshadhri2014} shows a matching lower bound of $\Omega(\epsilon^{-1} d \log \epsilon n)$ for monotonicity, that holds even for two-sided adaptive tests, while the communication complexity approach of \cite{BlaisRY2014} gives a (non-adaptive, two-sided)  $\Omega(d \log n)$ lower bound for convexity, separate convexity and Lipschitz.

Submodularity is testable for $d=2$ with $O(\log^2 n)$ queries \cite{ParnasRR04}; However, no non-trivial upper bound on the query complexity is known for submodularity in the case $d > 2$ and convexity in the case $d > 1$ under the Hamming distance and over general alphabets (although \cite{BMR2015} proves constant-query testability for 2D convexity over a binary alphabet). Under $L_1$-distance and for any $d$, it was shown in \cite{BermanRaskhodnikovaYaroslavtev2014} that convexity in $[n]^d$-arrays is testable with number of queries that depends only on $d$.

\paragraph{Pattern freeness} In \cite{BenEliezerKormanReichman2017}, it was shown that the property of (consecutive) pattern freeness, for a single forbidden pattern, is testable with $O(c_d / \epsilon)$ queries for any $d$. The proof, however, requires multiple sophisticated combinatorial observations and does not seem to translate to the case of a family of forbidden patterns discussed here.  

\subsection{Our results}
\label{subsec:contributions}
In this work, we present a generic approach to test \emph{all} $k$-local properties of $[n]^d$-arrays. 
Among other consequences, a simple special case of our result in the one-dimensional regime shows that the abundance of properties whose query complexity is $\Theta(\log n)$ is not a coincidence: in fact, \emph{any} $O(1)$-local property of one-dimensional arrays is testable with $O(\log n)$ queries, using a canonical binary search like querying scheme.
 
On the other hand, we prove a lower bound for testing local properties in $d > 1$ dimensions, showing that the query complexity of our test is optimal (for fixed $d$) among non-adaptive one-sided tests, even when restricted to alphabets of size that is polynomial in $n^d$. We also prove a lower bound for non-adaptive two-sided tests.

\subsubsection{Upper bounds}
Our first main result is an upper bound on the number of queries required to test any $k$-local property of $[n]^d$-arrays non-adaptively with one-sided error. 
The test is \emph{canonical} in a strong sense: The queries it makes depend on $n, d, k$, and (relatively weakly) on $\epsilon$; they do not depend on $\mathcal{P}$ or the alphabet $\Sigma$. 
In other words, it makes the same type of queries for all $k$-local properties of $[n]^d$-arrays over any finite (and not necessarily bounded-size) alphabet.

\begin{thm}
	\label{thm:strong_upper_bound}
	Let $2 \leq  k \leq n$ and $d \geq 1$ be integers, and let $\epsilon > 0$.
	Any $k$-local property $\mathcal{P}$ of $[n]^d$-arrays over any finite (and not necessarily bounded size) alphabet has a one-sided error non-adaptive $\epsilon$-test whose number of queries is
	\begin{itemize}
		\item $O(\frac{k}{\epsilon} \cdot \log{\frac{\epsilon n}{k}})$ for $d = 1$.
		\item $O(c^d \frac{k}{\epsilon^{1/d}}  \cdot n^{d-1})$ for $d > 1$.
	\end{itemize}
	Here, $c > 0$ is an absolute constant. The test chooses which queries to make based only on the values of $n, d, k, \epsilon$, and independently of the property $\mathcal{P}$ and the alphabet $\Sigma$.
\end{thm}
Note that we are interested here in the domain where $n$ is large and $d$ is considered a constant. Thus, we did not try to optimize the $c^{d}$ term in the second bullet, seeing that it is negligible compared to $n^{d-1}$ anyway. 

\paragraph{Running time}
The main drawback of our approach is the running time of the test, which is high in general. After making all of its queries, our test runs an \emph{inference} step, where it tries to evaluate (by enumerating over all relevant possibilities) whether a violation of the property must occur in view of the queries made, and reject if this is the case.

Without applying any property-specific considerations, the running time of the inference step is of order $|\Sigma|^{O(n^d)}$. However, for various specific properties of interest, such as monotonicity and 1D-convexity, it is not hard to make the running time of the inference step of the same order of magnitude as the query complexity. Moreover, in one dimension we can use dynamic programming to achieve running time that is significantly better than the naive one, but still much higher than the query complexity in general: $O(|\Sigma|^{O(k)} n)$. This works for any $k$-local property in one dimension; see the last part of Subsection \ref{subsec:upper_bound_1D} for more details.

\paragraph{Proximity oblivious test}
Interestingly, the behavior of the test depends quite minimally on $\epsilon$, and it can be modified very slightly to create a proximity oblivious test (POT) for any $k$-local property.
The useful notion of a POT, originally defined by Goldreich and Ron \cite{GoldreichRon2011}, refers to a test that does \emph{not} receive $\epsilon$ as an input, and whose success probability for an input not satisfying the property is a function of the Hamming distance of the input from the property.

\begin{thm}
	\label{thm:proximity_oblivious}
Fix $d > 0$. Any $k$-local property $\mathcal{P}$ of $[n]^d$-arrays over any finite (but not necessarily bounded size) alphabet has a one-sided error non-adaptive proximity oblivious test whose number of queries is $O(k \log(n/k))$ if $d=1$ and $O(k n^{d-1})$ if $d > 1$. For any input $A$ not satisfying $\mathcal{P}$, the rejection probability of $A$ is linear (for fixed $d$) in the Hamming distance of $A$ from $\mathcal{P}$. 
\end{thm} 

One can run $O(c_d/\epsilon)$ iterations of the POT to obtain a standard one-sided error non-adaptive test. The query complexity is $O(k\epsilon^{-1} \log(n/k))$ for $d = 1$ and $O(c_d k \epsilon^{-1} n^{d-1})$ for $d > 1$, where $c_d > 0$ depends only on $d$. Thus, the POT-based test is sometimes as good as the test of Theorem \ref{thm:strong_upper_bound} (specifically, for $d=1$ it matches the above bounds for almost the whole range of $\epsilon$ and $k$). In any case, the multiplicative overhead of the POT-based test is sublinear in $1/\epsilon$ across the whole range.

\paragraph{Type of queries}
In one dimension, many of the previously discussed properties, including, for example, monotonicity and Lipschitz continuity, are testable in $O(\log n)$ queries (see Subsection \ref{subsec:previous_results_local} for a more extensive discussion). Previously known tests for monotonicity and Lipschitz continuity make queries that resemble a binary search in some sense: these tests query pairs of entries of distance $2^i$ for multiple choices of $0 \leq i \leq \log n$. 

Our test continues the line of works using querying schemes roughly inspired by binary search. The test queries structures that can be viewed, intuitively, as $L_{\infty}$-\emph{spheres} of different sizes in $[n]^d$. For this purpose, an $L_{\infty}$-sphere with radius $r$ and width $\ell$ in $[n]^d$ is a set $X_1 \times X_2 \times \ldots \times X_d \subseteq [n]^d$, where each $X_i$ is a union of intervals of the form $[a_i, a_{i}+1, \ldots, a_i+\ell-2, a_{i}+\ell-1] \cup [b_i-\ell+2, b_{i}-\ell+3, \ldots, b_i - 1, b_i]$, and $b_i - a_i \in \{r, r+1\}$ for any $i \in [d]$. More specifically, our test for $k$-local properties queries spheres with width $k-1$ and radius of order $2^i$ for different values of $i$. In the simple special case where $d=1$ and $k=2$, this is very similar to the querying scheme mentioned in the previous paragraph. 

\paragraph{Implications}
In one dimension, the query complexity of the test matches the best known upper bounds (and, in some regimes, refines the dependence on $\epsilon$) for several previously investigated properties including monotonicity, Lipschitz continuity and convexity. For monotonicity of $k$-th order derivatives, which is $(k+1)$-local, it proves the first sublinear upper bound on the query complexity: $O(k \log n)$; in comparison, the best known lower bound \cite{BlaisRY2014} is $\Omega(\log n)$. 

For pattern matching type properties in 1D arrays (including applications in computational biology and other areas), our approach gives a property- and alphabet-independent upper bound of $O(k \log n)$ on the query complexity, with essentially optimal dependence on $\epsilon$ as well. 
Previously known approaches for testing such properties, like the regular languages testing approach \cite{Alon2001}, yield tests whose query complexity is dependent on the family of forbidden patterns considered, whose size might be exponential in the locality parameter $k$. Our approach, on the other hand, requires an $O(\log n)$ ``overhead'', but its query complexity is independent of the size of the forbidden family discussed. Instead, the dependence in $k$ is linear. 

In multiple dimensions, our approach is far from tight for well-understood properties such as monotonicity and Lipschitz continuity, whose query complexity is known to be $\Theta(d \log n)$ (in comparison, our approach yields an $O(n^{d-1})$ type bound). However, for testing of other properties like convexity (for $d > 1$) and submodularity (for $d > 2$) in $[n]^d$-arrays, no non-trivial upper bounds on the query complexity are known over general alphabets, so our upper bound of $O(n^{d-1})$ is the first such bound. While we do not believe this bound is tight in general, this might be a first step towards the development of new tools for efficiently testing such properties.

\paragraph{Sketching for testing} The fact that the queries made are completely independent of the property  suggests the following sketching technique allowing for ``testing in retrospect'': Given $\epsilon$ and $k$ in advance, we make all queries of the generic $\epsilon$-test for $k$-local properties in ``real time'', and store them for postprocessing. This is suitable, for example, in cases where we have limited access to a large input for a limited amount of time (e.g. when reading the input requires specialized expensive machinery), but the postprocessing time is not an issue. Note that for this approach we do not need to know the property of interest in advance.

\subsubsection{Lower bounds}
Our next main result is a tight lower bound for non-adaptive one-sided error testing of local properties, that applies for any $d$, and is tight for any fixed $d$ satisfying $d > 1$.

\begin{thm}[One-sided tests]
\label{thm:main_lower_bound}
Let $d \geq 1$ and $n \geq k \geq 2$ be integers, and let $d/n < \epsilon < 1$.
There exists a $k$-local property $\mathcal{P}$ of $[n]^d$-arrays over an alphabet $\Sigma$ of size $n^{O(d)}$, so that any non-adaptive one-sided error $\epsilon$-test for $\mathcal{P}$ requires $\Omega \left(\min \left\{ \frac{k}{d\epsilon^{1/d}}  \cdot n^{d-1}, n^d \right\} \right)$ queries.
\end{thm}
Note that the size of the alphabet in Theorem \ref{thm:main_lower_bound} is only polynomial in the input size.
We also prove a lower bound for non-adaptive two-sided tests; here the dependence in $|\Sigma|$ is exponential.
\begin{thm}[Two-sided tests]
\label{thm:two_sided_lower_bound}
Let $d \geq 1$ and $n \geq k \geq 2$ be integers, and let $d/n < \epsilon < 1$. There exists a $k$-local property $\mathcal{P}$ of $[n]^d$-arrays over an alphabet $\Sigma$ of size $2^{O(n^d)}$, so that any non-adaptive two-sided error $\epsilon$-test for $\mathcal{P}$ requires $\Omega \left(\min \left\{ \frac{\sqrt{k}}{d\epsilon^{(d+1)/2d}}  \cdot n^{(d-1)/2}, n^d \right\} \right)$ queries.
\end{thm}
For fixed $d > 1$, the lower bound for one-sided tests matches the upper bound from Theorem \ref{thm:strong_upper_bound} across the whole range of $\epsilon$ and $k$. For $d=1$ the bound obtained here is $\Omega(k / \epsilon)$, which is tight up to a $\log{n}$ factor. Note the threshold behavior occurring at $k / \epsilon^{1/d} = \Theta(n)$: When $k / \epsilon^{1/d} = o(n)$, the upper bound of Theorem \ref{thm:strong_upper_bound} implies that any $k$-local property is $\epsilon$-testable with a sublinear number of queries, while for $k / \epsilon^{1/d} = \Omega(n)$, the property of Theorem \ref{thm:main_lower_bound} requires $\Omega(n^d)$ non-adaptive one-sided queries to test.

From Theorem \ref{thm:two_sided_lower_bound} we conclude that the improvement in query complexity obtained by two-sided error non-adaptive tests is at most quadratic in the worst case; specifically, there exists a $2$-local property requiring $n^{\Omega(d-1)}$ queries to test by two-sided non-adaptive tests.
\subsection{Proof ideas and techniques}
Here we present the main ideas of our proofs in an informal way, starting with the upper bound. For simplicity, we stick to the one-dimensional case, and assume that $\epsilon$ is fixed and $k = o(n)$. 

\subsubsection{Upper bound for 1D}
\label{subsec:upper_bound_1D}
Suppose that $\mathcal{P} = \mathcal{P}(\mathcal{F})$ is a $k$-local property of $[n]^1$-arrays $A$ over an alphabet $\Sigma$, defined by the forbidden family $\mathcal{F}$. Let $S$ be a consecutive subarray of $A$ of length at least $2k-2$.
The \emph{boundary} of $S$ consists of the first $k-1$ elements and the last $k-1$ elements of $S$, and all other elements of $S$ are its \emph{interior}.
We call $S$ \emph{unrepairable} if one cannot make the array $S$ satisfy the property $\mathcal{P}$ without changing the value of at least one element in its boundary. Otherwise, $S$ is \emph{repairable}. Observe the following simple facts.
\begin{itemize}
	\item It suffices to only query the boundary elements of $S$ in order to determine whether $S$ is unrepairable. 
	\item If $S$ is unrepairable, then $A$ does not satisfy $\mathcal{P}$.
	\item If $S$ is repairable, then we can delete all forbidden patterns from $S$ by modifying only entries in its interior, without creating any new copies of forbidden patterns in $A$.
\end{itemize} 
We call the process of understanding whether $S$ is unrepairable using only its boundary elements \emph{inference}. Note that the inference step does not make any additional queries.

\paragraph{A simple sublinear test}
A first attempt at a generic test for local properties is the following: we query $\Theta(\sqrt{n})$ intervals in $[n]$, each containing exactly $k-1$ consecutive elements, including the intervals $\{1, \ldots, k-1\}$ and $\{n-k+2, \ldots, n\}$, where the distance between each two neighboring intervals is $\Theta(\sqrt{n})$. A \emph{block} is a subarray consisting of all elements in a pair of neighboring intervals and all elements between them.
The crucial observation is that at least one of the following must be true, for any array $A$ that is $\epsilon$-far from $\mathcal{P}$ (recall that $\epsilon$ is fixed).
\begin{itemize}
	\item At least one of the blocks is unrepairable.
	\item At least $\Omega(\sqrt{n})$ of the blocks do not satisfy $\mathcal{P}$.
\end{itemize}
Indeed, if the first condition does not hold, then one can make $A$ satisfy $\mathcal{P}$ by only changing elements in the interiors of blocks that do not satisfy $\mathcal{P}$. Seeing that $A$ is $\epsilon$-far from $\mathcal{P}$ and that we do not need to modify elements in the interiors of blocks that satisfy $\mathcal{P}$, this implies that at least $\Omega(\sqrt{n})$  of the blocks do not satisfy $\mathcal{P}$.

Now we are ready to present the test: We query all $O(k\sqrt{n})$ elements of all intervals, and additionally, all $O(\sqrt{n})$ elements of $O(1)$ blocks. Querying all elements of all intervals suffices to determine (with probability $1$) whether one of the blocks is unrepairable. If $A$ is $\epsilon$-far from $\mathcal{P}$ and does not contain unrepairable blocks, querying $O(1)$ full blocks will catch at least one block not satisfying $\mathcal{P}$ with constant probability, as desired. 

For more details, see Section \ref{sec:weak_test} and the preliminary Section \ref{sec:grid_structure} that prepares the required infrastructure.

\paragraph{The optimal test}
Improving the query complexity requires us to construct a system of \emph{grids} -- which are merely subsets of $[n]$ -- inspired by the behavior of binary search. In comparison, the approach of the previous test is essentially to work with a single grid.

The first (and coarsest) grid contains only the first $k-1$ elements and the last $k - 1$ elements of $[n]$. In other words, it is equal to $\{1, \ldots, k-1, n-k+2, \ldots, n\}$.
The second grid refines the first grid -- that is, it contains all elements of the first grid -- and additionally, it contains $k-1$ consecutive elements whose center is $n/2$ (whenever needed, rounding can be done rather arbitrarily). 
We continue with the construction of grids recursively: To construct grid number $i+1$, we take grid number $i$ and add $k-1$ elements in the middle of each block of grid $i$ (blocks are defined as before). Note that the length of blocks is roughly halved with each iteration. We stop the recursive construction when the length of all of the blocks becomes no bigger than $ck$, where $c \geq 2$ is an absolute constant. 

For each block $B$ in grid number $i > 1$, we define its \emph{parent}, denoted $\Par(B)$, as the unique block in interval $i-1$ containing it. A block $B$ in the system of grid is \emph{maximally unrepairable} if it is unrepairable, and all blocks (of all grids in the system) strictly containing it are repairable. It is not hard to see that different maximally unrepairable blocks have disjoint interiors.

The main observation now is that in order to make $A$ satisfy $\mathcal{P}$, it suffices to only modify entries in the interiors of parents of maximally unrepairable blocks. If $A$ is $\epsilon$-far from $\mathcal{P}$, then the total length of these parents must therefore be $\Omega(n)$ (for constant $\epsilon$). However, since the length of $\Par(B)$ is roughly twice the length of $B$, we conclude that the total length of all maximally unrepairable blocks is $\Omega(n)$. 

With this in hand, it can be verified that the following test has constant success probability. For each grid in the system, we pick one block of the grid uniformly at random, and query all entries of its boundary. Additionally, for the finest grid (whose block length is $O(k)$), we also query all interior elements of the picked block.

For more details, see Section \ref{sec:strong_test} (which builds on the infrastructure of Section \ref{sec:grid_structure}).

\paragraph{Running time in 1D}
We now show that the running time of the inference step for a block of length $m$ is $m |\Sigma|^{O(k)}$. Summing over all block lengths, this would imply that the total running time of the test is $n |\Sigma|^{O(k)}$. 

The proof uses dynamic programming. Let $S$ be an array of length $m$ over $\Sigma$, and assume that $S(1), S(2), \ldots, S(k-1)$ and $S(m-k+2), \ldots, S(m)$ are all known. For each ``level'' from 1 to $m-k+1$, we keep a boolean predicate for each of the $|\Sigma|^{k}$ possible patterns of length $k$ over $\Sigma$. These predicates are calculated as follows.

\begin{itemize}
	\item In the first level, the predicate of $\sigma = (\sigma_1, \ldots, \sigma_k)$ evaluates to TRUE if $S(1) = \sigma_1, \ldots, S(k-1) = \sigma_{k-1}$, and additionally, $\sigma \notin \mathcal{F}$, that is, $\sigma$ is not a forbidden pattern. Otherwise, the predicate of $\sigma$ is set to FALSE.
	\item For $i = 2$ to $m - k + 1$, the predicate of $\sigma = (\sigma_1, \ldots, \sigma_k)$ in level $i$ evaluates to TRUE if and only if
	\begin{enumerate}
		\item $\sigma \notin \mathcal{F}$.
		\item there exists $\sigma' = (\sigma'_0, \sigma_1, \ldots, \sigma_{k-1})$ that evaluates to TRUE in level $i-1$.
	\end{enumerate}
	\item Finally, the predicates in level $m - k + 1$ are modified as follows: for all $\sigma = (\sigma_1, \ldots, \sigma_k)$ so that $\sigma_j \neq S(m-k+j)$ for some $j \geq 2$, we set the predicate of $\sigma$ to FALSE.
\end{itemize}
It is not hard to see that $S$ is unrepairable if and only if all predicates at level $m-k+1$ are FALSE. The running time is $O(m |\Sigma|^{cd})$ for a suitable constant $c > 0$.

\paragraph{Generalization to higher dimensions}
The generalization to higher dimensions is relatively straightforward; the main difference is that the boundary of blocks now is much larger: blocks of size $m \times \ldots \times m$ have boundary of size $O(kdm^{d-1})$. Thus, essentially the same proof as above (with suitable adaptations of the definitions) yields a test with query complexity $O(kdn^{d-1})$ for constant $\epsilon$. For the running time, we can no longer use dynamic programming; using the naive approach of enumerating over all possible interior elements of a block, we get that the inference time for a block of size $m \times \ldots \times m$ is $|\Sigma|^{O(m^d)}$, making the total running time of the test $|\Sigma|^{O(n^d)}$.

\subsubsection{Lower bound}
The property $\mathcal{P}$ underlying our lower bound construction, fully described in Section \ref{sec:lower_bound}, consists of $[n]^d$-arrays $A$ over $\Sigma$ satisfying all of the following properties. Here we provide a construction over alphabet size $2^{O(n^d)}$, but in Section \ref{section:lower_shrinking} we show how a simple modification of the property can be conducted in order to decrease the size of the required alphabet to $n^{O(d)}$ for the proof of Theorem \ref{thm:main_lower_bound} (unfortunately, for the proof of Theorem \ref{thm:two_sided_lower_bound} this modification does not work).
\begin{itemize}
	\item The alphabet $\Sigma$ is of the form $[n]^d \times [n]^d \times 2^{[2n^{d-1}]}$, where $2^X$ is the \emph{power set} of a set $X$. The value of $A$ in entry $x \in [n]^d$ is represented as a tuple $A_1(x), A_2(x), A_3(x)$.  We view $A_1(x), A_2(x)$ as \emph{pointers} emanating from $x$.
	\item For every $x \in [n]^d$ we require $A_1(x) = x$. That is, $A_1$ points to the location of the element itself.
	\item There exists a special location $\ell = (\ell_1, \ldots, \ell_d) \in [n]^d$ so that all $x \in [n]^d$ point to $\ell$ with their second pointer, that is, $A_2(x) = \ell$. We call this location the \emph{lower center of gravity}. We define an \emph{upper center of gravity} as $u = (\ell_1+1, \ell_2, \ldots, \ell_d)$.
	\item A \emph{floor} entry $x = (x_1, \ldots, x_d) \in [n]^d$ satisfies $x_1 = 1$ and a \emph{ceiling} entry satisfies $x_1 = n$. For each such floor or ceiling entry $x$, we pick $A_3(x)$ to be a singleton (i.e., a set with one element). 
	\item For each floor element $x$, there exists a path $\Gamma_x$ from $x$ to the lower center of gravity $\ell$. Similarly, for each ceiling element $y$, there is a path $\Gamma_y$ directed towards the upper center of gravity, $u$. In both cases, the path is of length $O(nd)$. The structure of the path depends only on its start and end points (so depends only on $x$ and $\ell$ in the first case, and $y$ and $u$ in the second case).
	\item The $A_3$-data ``flows'' to the center of gravity through paths. Formally, the $A_3$-set of each entry $y$ that is not a floor or ceiling entry is required to be equal to the union $\bigcup_{x \colon y \in \Gamma_x} A_3(x)$. In other words, the data in each location in $[n]^d$ is an ``aggregation'' of the data flowing in all paths that intersect it.
	\item Finally, we require that $A_3(u) = A_3(\ell)$. 
\end{itemize}

While $\mathcal{P}$ was defined above in global terms, we show that it is actually a $2$-local property, that is, all conditions specified here can be written in a $2$-local way.

To prove the lower bound, we follow Yao's minimax principle \cite{Yao77}, defining a distribution of arrays satisfying $\mathcal{P}$, and a distribution of arrays which are $\Omega(1)$-far from satisfying $\mathcal{P}$, so that a large number of queries is required to distinguish between the distributions.
  
As positive examples, we take a collection of arrays $A$ satisfying the property, and require that all the singletons in the floor are pairwise disjoint. 
For negative examples (that are $1/4$-far from $\mathcal{P}$), we consider a collection of arrays satisfying all of the above requirements other than the last. Instead, all singletons in the floor and the ceiling are pairwise disjoint (so in particular, $A_3(\ell) \cap A_3(u) = \emptyset$).

We show that for any given $x \in [n]^d$, the expected size of $A_3(x)$ over each of the distributions is $O(d)$. For the one-sided error case, it is shown that one needs to know the values of at least $\Omega(n^{d-1})$ singletons to be able to distinguish between positive and negative examples with one sided error, implying that $\Omega(n^{d-1} / d)$ queries are required to reject negative examples with constant probability. 

For the two sided error case, the argument is inspired by the birthday paradox. Very loosely speaking, it follows from the fact that, given two unknown unordered sets $A$ and $B$ of size $n$, one has to make $\Omega(\sqrt{n})$ queries to distinguish between the case that $A = B$ and the case that $A \cap B = \emptyset$.

\subsection{Other related work}
\label{subsec:related}

This subsection complements Subsection \ref{subsec:previous_results_local}, presenting other related previous works that were not mentioned above.

\paragraph{General results in property testing}
This paper adds to the growing list of general characterization results in property testing of strings, images, and multi-dimensional arrays; see \cite{AlonBenEliezerFischer2017, BenEliezerFischer2018} and the references within for characterization-type results in these domains, mostly over a fixed size alphabet. 
In particular, for strings, it was shown by Alon et al. \cite{Alon2001} that any local property over a fixed size alphabet is constant-query testable, and this paper shows that an overhead of at most $O(\log n)$ is required when the alphabet size is unbounded. 

\paragraph{Hyperfiniteness}
A graph is hyperfinite if, roughly speaking, it can be decomposed into constant size connected components by deleting only a small constant fraction of the edges.
Newman and Sohler \cite{NewmanSohler2013} investigated the problem of testing in hyperfinite graphs,
showing that any property of hyperfinite bounded degree 
graphs is testable with a constant number of queries. While the graph with which we (implicitly) work -- the hypergrid graph, whose vertices are in $[n]^d$ and two vertices are neighbors if they differ by $1$ in one coordinate -- is a hyperfinite bounded degree graph (for constant $d$), the results of \cite{NewmanSohler2013} are incomparable to ours. Indeed, in our case the vertices are inherently ordered, and it does not make sense to allow adding edges between vertices that are not neighbors (as entries of $[n]^d$), unlike the case in \cite{NewmanSohler2013}, where one may add or remove edges arbitrarily between any two vertices. Still, the hyperfiniteness of our graph seems to serve as a major reason that local properties have sublinear tests.

\paragraph{Block tests for image properties}
The works of Berman, Murzabulatov and Raskhodnikova \cite{BMR2015} and Korman, Reichman, and the author \cite{BenEliezerKormanReichman2017} on testing of image properties (that is, on visual properties of 2D arrays) show that tests based on querying large consecutive blocks are useful for image property testing.
In this work, the general queries we make are quite different: we query the \emph{boundaries} of blocks of different sizes, so the queries are \emph{spherical}, in the sense that a block can be seen as a ball in the $L_{\infty}$-metric on vectors in $[n]^d$, while its boundary can be be seen as the (width-$k$) sphere surrounding this ball. This introduces a new type of queries shown to be useful for image property testing.

\subsection{Discussion}
\label{subsec:discussion}
\paragraph{Small alphabets}
The results in this work are alphabet independent, and in particular, they work for alphabets over any size.
An intriguing direction of research is to understand whether one can obtain more efficient general testability results for local properties of multi-dimensional arrays over smaller alphabets; this line of research has been conducted for specific properties of interest, like monotonicity and convexity \cite{Belovs2018, PalavoorRV2018}. Note that the one-sided non-adaptive lower bound we prove here can be adapted to yield a $|\Sigma|^{\Omega(1)}$ lower bound for testing local properties over alphabets $\Sigma$ of size smaller than $n^d$.

The most interesting special case is that of constant-sized (and in particular, binary) alphabets. Here, no lower bounds that depend on $n$ are known. 
For the case $d = 1$, it is known that all $O(1)$-local properties are constant query testable; this follows from a result of Alon et al. \cite{Alon2001}, who showed that any regular language is constant-query testable.
However, it is not known whether an analogous statement holds in higher dimensions. That is, for any $d > 1$,
the question whether all $k$-local properties of $[n]^d$-arrays over $\{0,1\}$ are $\epsilon$-testable with query complexity that depends only on $d$, $k$, and $\epsilon$, first raised in \cite{BenEliezerKormanReichman2017} (see also \cite{AlonBenEliezerFischer2017}), remains an intriguing open question. We believe that positive results in this front might also shed light on the question of obtaining more efficient inference for large classes of properties, especially over small alphabets.  

\paragraph{Does adaptivity help?} 
This work does not provide any lower bounds for adaptive tests, and it will be interesting to do so; previously investigated properties likes monotonicity yield an $\Omega(d \log n)$ lower bound \cite{BlaisRY2014, ChakrabartySeshadhri2014}, and we believe that ``data flow'' type properties, somewhat similar to our lower bound constructions, can provide instances of $2$-local properties that require at least $n^c$ queries, for some constant $c \leq 1$, for the adaptive two-sided case.

However, it is not clear whether better lower bounds (even bounds of the type $\Omega(n^{1+c})$) exist. It will be very interesting to prove better upper and lower bounds for testing local properties. Our conjecture is that \emph{any} $2$-local property is testable in $n^{1+o(1)} g(d)$ queries (where $g(d)$ depends only on $d$), but proving a statement of this type might be very difficult.

\paragraph{Using the unrepairability framework in other contexts}
In this work we show that the concept of unrepairability allows to unify and reprove many property testing results on one-dimensional arrays. What about multi-dimensional arrays? for example, can one generalize the currently known proofs for ``bounded derivative'' properties (including monotonicity and Lipschitz continuity) in $d$ dimensions to a larger class of local properties? 

\paragraph{Inference}
As mentioned in Subsection \ref{subsec:upper_bound_1D}, our test queries boundaries of block-like structures, and later infers whether each block is  \emph{unrepairable} (recall the definition from Subsection \ref{subsec:upper_bound_1D}). The inference takes place without making any additional queries, and is based only on the property $\mathcal{P}$, the alphabet $\Sigma$, and the values of $A$ in the boundary of the block.  

The running time of the inference step is very large in general (although, as we have seen, in the 1D case it can be significantly improved using dynamic programming). The naive way to run the inference is by enumerating over all possible ways to fill the interior of the block, and checking whether each such possibility is indeed $\mathcal{F}$-free. The running time of this method is of order $|\Sigma|^{O(n^d)}$ in general for $d > 1$, and is exponential in $n$ even if $|\Sigma| = 2$.  

However, for many natural properties, inference can be done much more efficiently. For example, in monotonicity testing, the inference amounts to checking that no pair of boundary entries violates the monotonicity. The lower bound constructions from Sections \ref{sec:lower_bound} and \ref{section:lower_shrinking} depict other properties where inference is efficient: it is not hard to show that the running time of inference in both cases is $O(k \epsilon^{-1/d} n^{d-1})$, which is sublinear in $n^d$ for a wide range of parameters. 

Thus, we believe that understanding inference better -- including tasks  such as characterizing properties in which inference can be done efficiently, and understanding the inference time of specific properties of interest -- would be an interesting direction for future research.

\subsection{Property testing notation}
\label{subsec:notation}
The property testing notation we use along the paper is standard. Given a property $\mathcal{P}$ of $[n]^d$-arrays over $\Sigma$, a proximity parameter $\epsilon > 0$, and query access to an unknown $[n]^d$-array $A$, a \emph{two-sided error $\epsilon$-test} must accept $A$ with probability at least $2/3$ if $A$ satisfies $\mathcal{P}$, and reject with probability $2/3$ if $A$ is \emph{$\epsilon$-far} from $\mathcal{P}$ (meaning that the relative \emph{Hamming distance} of $A$ from $\mathcal{P}$ is at least $\epsilon$, that is, we need to modify at least $\epsilon n^d$ values in $A$ to make it satisfy $\mathcal{P}$). 
A \emph{one-sided error test} is defined similarly, but it must accept if $A$ satisfies $\mathcal{P}$. A test is \emph{non-adaptive} if it makes all of its queries in advance (prior to receiving any of the queried values), and \emph{adaptive} otherwise. 

\paragraph{Organization}
In Sections \ref{sec:grid_structure}, \ref{sec:weak_test} and \ref{sec:strong_test} we prove the upper bounds: Section \ref{sec:grid_structure} is devoted to the infrastructure needed for the proof, Section \ref{sec:weak_test} presents a simple but non-optimal test, and finally, Section \ref{sec:strong_test} presents the optimal test and proves Theorems \ref{thm:strong_upper_bound} and \ref{thm:proximity_oblivious}.

In Sections \ref{sec:lower_bound} and \ref{section:lower_shrinking} we prove the lower bounds: the former is devoted to proving the bounds with an alphabet size that of exponential size, while the latter shows to shrink the alphabet in the proof of Theorem \ref{thm:main_lower_bound}.
%%%%%%%%%%%%%%%%%%%%%%%%%%%%%%%%%%%%%%%%%%%%%%%
\section{The grid structure}
\label{sec:grid_structure}
%%%%%%%%%%%%%%%%%%%%%%%%%%%%%%%%%%%%%%%%%%%%%%%

In this section we present the grid-like structure in $[n]^d$ that we utilize for our tests. 

\begin{definition}[Interval partition]
A subset $I \subseteq [n]$ is an \emph{interval} if its elements are consecutive, that is, if $I = \{x, x+1, \ldots, x+y\}$ for some $x \in [n]$ and $y \geq 0$. For any $\ell \geq 0$, we denote the set of the smallest $\ell$ elements of $I$ by $I[\colon\ell]$ and also define $I[\ell+1\colon] = I \setminus I[\colon\ell]$. In the degenerate case that $|I| < \ell$, we define $I[\colon\ell]$ to be equal to $I$.

For $1 \leq w \leq n$, an \emph{$(n, w)$-interval partition} is a partition of $[n]$ into a collection of disjoint intervals $\mathcal{I} = (I_1, \ldots, I_t)$ where the number of elements in each interval $I_i$ is either $w$ or $w+1$, and for any $ i < j$, all elements of $I_i$ are smaller than those in $I_j$. 
\end{definition}

\begin{lem}
	\label{lem:interval_partition}
For any positive integer $n$ and $0 \leq i \leq \log n$, there exists an $(n, \lfloor n/2^i \rfloor)$-interval partition $\mathcal{I}_i$ containing exactly $2^i$ intervals, so that the family $\{\mathcal{I}\}_{i=0}^{\lfloor \log n \rfloor}$ satisfies the following. For any $i > j$ and interval $I \in \mathcal{I}_i$, there exists an interval $I' \in \mathcal{I}_j$ satisfying $I \subseteq I'$.  
\end{lem}
\begin{proof}
For any $i$ define $n_i = \lfloor n / 2^i \rfloor$; observe that $n_0 = n$ and $n_{i+1} = \lfloor n_i / 2 \rfloor$ for any $i$.
We prove the lemma by induction on $i$, starting by defining $\mathcal{I}_0 = ([n])$.
Given $\mathcal{I}_i = (I_1^i, \ldots, I_{2^i}^i)$ in which all intervals are of length $n_i$ or $n_i+1$, we define $\mathcal{I}_{i+1}$ as follows. Each $I_j^i \in \mathcal{I}_i$, is decomposed into two intervals $I_{2j-1}^{i+1}, I_{2j}^{i+1}$ where $|I_{2j-1}^{i+1}|, |I_{2j}^{i+1}| \in \{n_{i+1}, n_{i+1}+1\}$, and all elements of $I_{2j-1}^{i+1}$ are smaller than all elements of $I_{2j}^{i+1}$; observe that such a decomposition is indeed always possible. Now define $\mathcal{I}_{i+1} = (I_1^{i+1}, \ldots, I_{2^{i+1}}^{i+1})$. Clearly, the intervals of $\mathcal{I}_{i+1}$ satisfy the last condition of the lemma.
\end{proof}
In particular, we conclude that for any positive integer $w$ and any $n \geq w$ there exists an integer $w/2 \leq w' \leq w$ for which an $(n,w')$-interval partition exists. 

\begin{definition}[$(n, d, k, w)$-grid]
Let $2 \leq w \leq n$ be integers for which an $(n, w)$-interval partition $\mathcal{I} = (I_1, \ldots, I_t)$ exists. 
For integers $2 \leq k \leq w$ and $d \geq 1$, the ($d$-dimensional) \emph{$(n, d, k, w)$-grid} induced by $\mathcal{I}$ is the set
$$
G = \left\{ (x_1, \ldots, x_d) \in [n]^d \ \bigg| \  \exists i \in [d]\ \text{ such that } x_i \in \bigcup_{j=1}^{t} I_j[\colon k-1] \right\}.
$$
We denote the family of all $(n, d, k, w)$-grids by $\mathcal{G}(n, d, k, w)$. As we have seen in Lemma \ref{lem:interval_partition}, the family $\mathcal{G}(n,d,k,w)$ is non-empty for any $w = \lfloor n / 2^i \rfloor$ satisfying $w \geq k$.  
\end{definition}

\begin{definition}[$G$-block, Boundary, Closure]
\label{def:block_boundary_closure}
Two tuples $x = (x_1, \ldots, x_d)$ and $y = (y_1, \ldots, y_d)$ in $[n]^d$ are considered \emph{neighbors} if $\sum_{i=1}^{d} |x_i - y_i| = 1$. Given a grid $G \in \mathcal{G}(n, d, k, w)$, consider the neighborhood graph of non-grid entries, i.e., the graph whose set of vertices is $V = [n]^d \setminus G$ and two entries are connected if they are neighbors. A $G$-\emph{block} $B$ is a connected component of this graph, and the \emph{closure} of $B$ is
$$
\overline{B} = \left\{(x_1, \ldots, x_d) \in [n]^d \ \bigg|\  \exists (y_1, \ldots, y_d) \in B \text{ such that }  \forall i \in [d] \  |x_i - y_i| < k \right\}.
$$ 
Note that $B \subseteq \overline{B}$. Define the \emph{boundary} of the block $B$ as $\partial B = \overline{B} \setminus B$.
\end{definition}

The above notions can naturally be defined with Cartesian products. 
Recall that the Cartesian product of sets $X_1, \ldots, X_d$, denoted $\prod_{j=1}^{d} X_j$ or $X_1 \times \ldots \times X_j$, is the set of all tuples $(x_1, \ldots, x_d)$ with $x_j \in X_j$ for any $j \in [d]$. 
Let $G \in \mathcal{G}(n, d, k, w)$ be the grid induced by the interval partition $\mathcal{I} = (I_1, \ldots, I_t)$.
It is not difficult to verify that any $G$-block $B$ can be defined as a Cartesian product $B = \prod_{j=1}^{d} I_{i_j}[k\colon ]$ for some intervals $I_{i_1}, \ldots, I_{i_d} \in \mathcal{I}$ (not necessarily different). 

$\overline{B}$ and $\partial B$ can also be defined accordingly, as we detail next. For $k$ as above, define $\overline{I_i} = I_{i} \cup I_{i+1}[\colon k-1]$ for any $1 \leq i \leq t$, where we take $I_{t+1} = \emptyset$ for consistency. Also define $\partial I_i = \overline{I_i} \setminus I_i[k\colon ] = I_i[\colon k-1] \cup I_{i+1}[\colon k-1]$. With these in hand, we have
\begin{align}
\label{eqn:explicit_representation}
B = \prod_{j=1}^d I_{i_j}[k\colon]\ ; \hspace{0.6cm}
\overline{B} = \prod_{j=1}^d \overline{I_{i_j}} \ ; \hspace{0.6cm}
\partial B = \bigcup_{j=1}^{d} \overline{I_{i_1}} \times \ldots \times \overline{I_{i_{j-1}}} \times \partial I_{i_j} \times \overline{I_{i_{j+1}}} \times \ldots \times \overline{I_{i_d}}
\end{align} 
Recall that $|I_{i_j}| \in \{w, w+1\}$ for any $j$, implying that $\big|I_{i_j}[k:]\big| \leq w+2-k$ and $\big|\overline{I_{i_j}}\big| \leq w+k$. Also note that $\big|\partial I_{i_j}\big| \leq 2(k-1)$. Thus,
\begin{align}
\label{eqn:block_size}
|B| \leq (w+2-k)^d \ ; \hspace{1cm}
|\overline{B}| \leq (w+k)^d \ ; \hspace{1cm} 
|\partial B| \leq 2d (k-1) \cdot (w+k)^{d-1}\ ,
\end{align} 
where the inequality on $|\partial B|$ holds since each set in the union expression in \eqref{eqn:explicit_representation} is of size at most $(2k-2) (w+k)^{d-1}$.

The following observation is a direct consequence of \eqref{eqn:explicit_representation}.
\begin{obs}
\label{obs:boundary_contained}
Let $G \in \mathcal{G}(n, d, k, w)$. The boundary of any $G$-block is contained in $G$.
\end{obs}

\begin{lem}
\label{lem:any_copy_contained}
For any $G \in \mathcal{G}(n, d, k, w)$, any width-$k$ subarray of an $[n]^d$-array intersects exactly one $G$-block $B$. Moreover, the subarray is contained in $\overline{B}$.
\end{lem}
%The proof of the lemma is somewhat tedious. Still, we provide the full details of the proof to help the reader digest the above notation.
\begin{proof}
Let $\mathcal{I} = (I_1, \ldots, I_t)$ be the interval partition inducing $G$.
Suppose that the subarray $S$ is in location $(a_1 ,\ldots, a_d)$ where $a_j \in I_{i_j}$ for some $i_1, \ldots, i_d$ not necessarily distinct. In other words, the set of entries in $S$ is $\prod_{j=1}^{d} S_j$ where $S_j = \{a_j, a_j+1, \ldots, a_j + k-1\}$ for any $j \in [d]$.
We argue that $S$ is contained in $\overline{B}$, where $B = I_{i_1}[k\colon ] \times \ldots \times I_{i_d}[k\colon ]$: The fact that $a_j \in I_{i_j}$ implies that $a_j+1, \ldots, a_{j}+k-1 \in I_{i_j} \cup I_{i_j+1}[\colon k-1]$. It follows from \eqref{eqn:explicit_representation} that $S \subseteq \overline{B}$. 
From Observation \ref{obs:boundary_contained} we conclude that $S$ does not intersect any block other than $B$, and it remains to show that $S$ intersects $B$. Indeed, for any $1 \leq j \leq d$, the fact that $a_j \in I_{i_j}$ implies that one of the elements $a_j, \ldots, a_j + k-1$ must be contained in $I_{i_j}[k\colon ]$. Denoting this element by $b_j$, we conclude that $(b_1, \ldots, b_j) \in S \cap B$. 
\end{proof} 

\section{Testing with grid queries}
\label{sec:weak_test}
In this section we prove the following upper bound for all $k$-local properties; its proof serves as a warm-up towards proving the main upper bound of Theorem \ref{thm:strong_upper_bound}.

\begin{thm}
	\label{thm:weak_upper_bound}
	Any $k$-local property of $[n]^d$-arrays over any alphabet is $\epsilon$-testable with one-sided error using no more than
	$2(d+1) n^{d - \frac{d}{d+1}} k^{\frac{d}{d+1}} \epsilon^{-\frac{1}{d+1}}$ non-adaptive queries.
\end{thm} 
The upper bound of Theorem \ref{thm:weak_upper_bound} is sublinear in the size of the array as long as $k / \epsilon^{1/d} = o(n)$. 
The rest of the section is dedicated to the proof of the theorem. We may assume that $k \leq \epsilon^{1/d} n / 4$, as otherwise the expression in the statement of the theorem is larger than $n^d$ and the proof follows trivially by querying all $[n]^d$ entries of the given input array. Under this assumption, it holds that $2k \leq n^{d/(d+1)} k^{1/(d+1)}\epsilon^{1/(d+1)}$.

\begin{definition}[Unrepairable block]
\label{def:unrepairable}
Let $A$ be an $[n]^d$-array over $\Sigma$, and let $G \in \mathcal{G}(n, d, k, w)$. A $G$-block $B$ is \emph{$(\mathcal{P}, A)$-unrepairable} (or simply \emph{unrepairable}, if $\mathcal{P}$ and $A$ are clear from context) if any $[n]^d$-array $A'$ over $\Sigma$ that satisfies $A'(x) = A(x)$ for any $x \in \partial B$, including the case $A' = A$, contains an $\mathcal{F}$-copy in $\overline{B}$. Otherwise, the block $B$ is said to be $(\mathcal{P}, A)$-\emph{repairable}. 
\end{definition}
Note that the (un)repairability of a block $B$ is determined solely by the values of $A$ on $\partial B$, and that an unrepairable block always contains an $\mathcal{F}$-copy. These two facts inspire the following lemma, which serves as the conceptual core behind the test of Theorem \ref{thm:weak_upper_bound}.

\begin{lem}
\label{lem:unrepairability1}
Suppose that $A$ is an $[n]^d$-array that is $\epsilon$-far from satisfying a $k$-local property $\mathcal{P}(\mathcal{F})$, and let $G \in \mathcal{G}(n, d, k, w)$ where $w \geq k$. Then at least one of the following holds.
\begin{itemize}
	\item There exists a $(\mathcal{P}, A)$-unrepairable $G$-block.
	\item For at least an $\epsilon$-fraction of the $G$-blocks $B$, there is an $\mathcal{F}$-copy in $\overline{B}$.
\end{itemize}
\end{lem}

\begin{proof}
Suppose that the first condition does not hold, that is, all ${G}$-blocks are $(\mathcal{P}, A)$-repairable. By Lemma \ref{lem:any_copy_contained}, every $\mathcal{F}$-copy is contained in the closure of some $G$-block.

Let $\mathcal{C}$ denote the collection of all $G$-blocks $B$ such that $A$ contains an $\mathcal{F}$-copy in $\overline{B}$. By the repairability, the values of $A$ in each block $B \in \mathcal{C}$ can be modified so that after the modification, $A$ will not contain an $\mathcal{F}$-copy in $\overline{B}$. We stress that the modifications for each block $B$ appear only in $B$ itself and do not modify entries on the grid, so by Lemma \ref{lem:any_copy_contained}, they cannot create new $\mathcal{F}$-copies in the closure of other blocks. 

After applying all of the above modifications to $A$, we get an $\mathcal{F}$-free array, i.e., an array that satisfies $\mathcal{P}$. $A$ was initially $\epsilon$-far from $\mathcal{P}$, and the number of entries in each block is bounded by $(w+2-k)^d \leq w^d$, implying that at least an $\epsilon$-fraction of the blocks belong to $\mathcal{C}$.
\end{proof}

\begin{proof}[Proof of Theorem \ref{thm:weak_upper_bound}]
We may assume that $k^d / \epsilon \leq n^d/2$, otherwise our test may trivially query all $n^d$ entries of $A$.
Our (non-adaptive) test $T$ picks $W = \lfloor n^{d / (d+1)} k^{1/(d+1)} \epsilon^{1/(d+1)} \rfloor \geq 2k$, and an integer $w$ satisfying $k \leq W/2 \leq w \leq W$, for which an $(n,w)$-interval partition exists. $T$ now makes the following queries.
\begin{enumerate}
	\item $T$ queries all entries of an arbitrarily chosen grid $G \in \mathcal{G}(n, d, k, w)$. The number of entries in any grid is at most $d n^d (k-1) / w \leq 2 d n^{d - \frac{d}{d+1}} k^{\frac{d}{d+1}} \epsilon^{-\frac{1}{d+1}}$.
	\item $T$ chooses a collection $\mathcal{B}$ of $2/\epsilon$ $G$-blocks uniformly at random and queries all entries in these blocks. Since each block contains at most $(w+2-k)^d \leq W^d$ entries, the total number of queries is bounded by $2 W^d / \epsilon \leq 2 n^{d - \frac{d}{d+1}} k^{\frac{d}{d+1}} \epsilon^{-\frac{1}{d+1}}$. Note that the boundaries of all blocks are queried in the first step (since they are contained in the grid). Thus, for any block $B \in \mathcal{B}$, the test queries all entries of $\overline{B}$.
\end{enumerate}
The total number of queries in the above two steps is $2(d+1) n^{d - \frac{d}{d+1}} k^{\frac{d}{d+1}} \epsilon^{-\frac{1}{d+1}}$.

After querying all entries of the grid (and in particular, the whole boundaries of all of the blocks), $T$ can determine for every $G$-block $B$ whether it is $(\mathcal{P}, A)$-unrepairable or not. $T$ rejects if at least one of the blocks is unrepairable or if it found an $\mathcal{F}$-copy in $\overline{B}$ for some $B \in \mathcal{B}$, and accepts otherwise. The test has one-sided error, since an unrepairable block must contain an $\mathcal{F}$-copy. In view of Lemma \ref{lem:unrepairability1}, $T$ rejects arrays $A$ that are $\epsilon$-far from $\mathcal{P}$ with probability at least $2/3$: If $A$ satisfies the first condition of Lemma \ref{lem:unrepairability1}, then $T$ always rejects. If the second condition holds, the probability that none of the $2 / \epsilon$ closures $\overline{B}$ for $B \in \mathcal{B}$ contains an $\mathcal{F}$-copy is bounded by $(1 - \epsilon)^{2 / \epsilon} < e^{-2}$, so $T$ rejects with probability at least $1 - e^{-2} > 2/3$.
\end{proof}

%%%%%%%%%%%%%%%%%%%%%%%%%%%%%%%%%%%%%%%%%%%%%%%
\section{Systems of grids and testing with spherical queries}
\label{sec:strong_test}
%%%%%%%%%%%%%%%%%%%%%%%%%%%%%%%%%%%%%%%%%%%%%%%
In this section we prove Theorems \ref{thm:strong_upper_bound} and \ref{thm:proximity_oblivious}. We do so by considering a system of grids with varying block sizes, defined as follows.
\begin{definition}
Let $d > 0$ and $2 \leq k \leq w \leq n$ be integers. 
An \emph{$(n, d, k, w)$-system of grids} is an $(r+1)$-tuple $(G_0, G_1, \ldots, G_r)$ of grids, for $r(n, w) = \lfloor \log(n/w) \rfloor$, so that 
\begin{itemize}
	\item $G_{i} \in \mathcal{G}(n, d, k, \lfloor n / 2^{r-i} \rfloor)$ for any $0 \leq i \leq r$.
	\item $G_0 \supseteq G_1 \supseteq \ldots \supseteq G_r$ (as subsets of $[n]^d$). In particular, for any $i < j \leq r$, any $G_i$-block $B$ is contained in a $G_j$-block $B'$, and we say that $B'$ is an \emph{ancestor} of $B$. Specifically, the $G_{i+1}$-block containing $B$ is called the \emph{parent} of $B$ and denoted by $\Par(B)$. For the only $G_r$-block, $B_r$, we define $\Par(B_r)$ as the whole domain $[n]^d$.
\end{itemize}
\end{definition}
$r(n,w)$ was chosen so that $w \leq n / 2^{r} < 2w$, making $G_0$ a $\mathcal{G}(n, d, k, w')$-grid for $w \leq w' < 2w$.
As we shall see, when working with such a system, unrepairability of blocks can be handled in a query-efficient way. The following lemma asserts that such a system of grids exists for any suitable choice of parameters.

\begin{lem}
\label{lem:grid_system_existence}
An $(n, d, k, w)$-system of grids exists for all $d > 0$ and $2 \leq k \leq w \leq n$.
\end{lem}
\begin{proof}
Consider the family of interval partitions $\mathcal{I}_0, \ldots, \mathcal{I}_{\lfloor \log n \rfloor}$ obtained by Lemma \ref{lem:interval_partition}. For each $0 \leq i \leq r(n,w)$ define $G_i$ as the $(n,d,k,\lfloor n / 2^{r-i} \rfloor)$-grid induced by $\mathcal{I}_{r-i}$. It is not hard to verify that $(G_0, \ldots, G_r)$ satisfies all requirements of an $(n,d,k,w)$-system of grids.
%	
%The proof is by induction on $r(n,w)$; the case $r = 0$ is trivial, since here the grid system consists of a single grid $G_0 \in \mathcal{G}(n, d, k, n)$. For the induction step, let $n$ and $w$ with $r(n, w) > 0$, and note that $r(n, 2w) = r(n, w) - 1$. By the induction assumption there exists an $(n, d, k, 2w)$-system of grids, i.e., there exist sets $G'_0, \ldots, G'_{r(n,w)-1}$ satisfying $G'_0 \supseteq \ldots \supseteq G'_{r(n,w) - 1}$ and $G'_i \in \mathcal{G}(n, d, k, \lfloor n / 2^{r(n, w)-1-i} \rfloor)$ for any $0 \leq i \leq r(n, w) - 1$. We construct an $(n, d, k, w)$-system of grids $(G_0, G_1, \ldots, G_r)$ as follows. First, for any $i > 0$ we take $G_i = G'_{i-1} \in \mathcal{G}(n, d, k, \lfloor n / 2^{r(n, w)-i} \rfloor)$. To construct $G_0$, write $w_1 = \lfloor n / 2^{r(n, w)-1} \rfloor$ and $w_0 = \lfloor n / 2^{r(n, w)} \rfloor$, and observe that $w_0 = \lfloor w_1 / 2 \rfloor$. Let $(I_1, \ldots, I_t)$ be the $(n, w_1)$-interval partition inducing $G_1$, and partition each of the intervals $I_i$ into two shorter intervals -- $I'_i$ and $I''_{i}$ -- satisfying $|I'_i|, |I''_i| \in \{ w_0, w_0 + 1\}$. The interval partition $(I'_1, I''_1, \ldots, I'_t, I''_t)$ is an $(n, w_0)$-interval partition, and we take $G_0$ as the $\mathcal{G}(n, d, k, w)$-grid induced by it. By definition, we get that $G_0 \supseteq G_1$, and so $(G_0, \ldots, G_r)$ is a valid $(n, d, k, w)$-system.
\end{proof}

For the rest of the section, fix a $k$-local property $\mathcal{P}(\mathcal{F})$ of $[n]^d$-arrays over $\Sigma$, as well as an $[n]^d$-array $A$ over $\Sigma$. Consider an $(n, d, k, w)$-system of grids $(G_0, \ldots, G_r)$ constructed as described in the proof of Lemma \ref{lem:grid_system_existence}, where $w$ will be determined later. (For now it suffices to require, as usual, that $2 \leq k \leq w \leq n$.) 

We say that a $G_i$-block $B$ is a \emph{$(\mathcal{P}, A)$-witness} if one of the following holds.
\begin{itemize}
	\item $i = 0$ and the array $A$ contains an $\mathcal{F}$-copy in the closure $\overline{B}$.
	\item $ i > 0$ and $B$ is $(\mathcal{P}, A)$-unrepairable. 
\end{itemize}
Recall that the closure of unrepairable blocks cannot be $\mathcal{F}$-free, so the closure of any witness block contains an $\mathcal{F}$-copy. We say that a witness block $B$ is \emph{maximal} if all of its ancestors are not witnesses, that is, they are repairable.

\begin{obs}
\label{obs:maximal_witness}
	Any $(\mathcal{P}, A)$-witness is contained in a maximal $(\mathcal{P}, A)$-witness.
\end{obs}

We define the \emph{maximal witness family} $\mathcal{W}$ as the set of all maximal $(\mathcal{P}, A)$-witness blocks. Obviously, the blocks in $\mathcal{W}$ might come from different $G_i$'s 

\begin{obs}
$B_1 \cap B_2 = \emptyset$ for any two blocks $B_1, B_2 \in \mathcal{W}$.
\end{obs}

\begin{lem}
All $\mathcal{F}$-copies in $A$ are fully contained in $\bigcup_{B \in \mathcal{W}} \overline{B}$.
\end{lem}
\begin{proof}
Let $F$ be an $\mathcal{F}$-copy in $A$. By Lemma \ref{lem:any_copy_contained}, $F$ is contained in the closure of a unique $G_0$-block $B_F$; hence, $B_F$ is a $(\mathcal{P}, A)$-witness. From Observation \ref{obs:maximal_witness} we have $B_F \subseteq B'$ for some maximal $(\mathcal{P}, A)$-witness $B'$. We conclude that $F \in \overline{B_F} \subseteq \overline{B'}$.
\end{proof}

\begin{lem}
\label{lem:only_parents}
One can make $A$ satisfy $\mathcal{P}$ by only modifying entries of $A$ in $\bigcup_{B \in \mathcal{W}} {\Par(B)}$.
\end{lem}
\begin{proof}
Fix $B \in \mathcal{W}$. $B$ is a maximal $(\mathcal{P}, A)$-witness, so $\Par(B)$ is repairable.\footnote{Note that when $B = B_r$ is the maximal witness considered, $\Par(B)$ is $[n]^d$; the latter is repairable for any non-empty property.} Thus, One can make $\overline{\Par(B)}$ $\mathcal{F}$-free by only modifying entries inside $\Par(B)$. By Lemma \ref{lem:any_copy_contained}, width-$k$ subarrays that are not fully contained in $\overline{\Par(B)}$ are left unchanged. Therefore, this modification does not create any new $\mathcal{F}$-copies in $A$. Seeing that all $\mathcal{F}$-copies in $A$ are originally contained in $\bigcup_{B \in \mathcal{W}} \overline{B} \subseteq \bigcup_{B \in \mathcal{W}} \overline{\Par(B)}$, applying these modifications for all $B \in \mathcal{W}$ deletes all $\mathcal{F}$-copies in $A$ without creating new ones, so in the end of the process $A$ satisfies $\mathcal{P}$.
\end{proof}

We may assume that $k \leq \epsilon^{1/d} n / 10$, as otherwise the expression in the theorem is $\Omega(n^d)$. We choose $w = 2k$, working with an $(n, d, k, 2k)$-system of grids from now on. A very useful consequence of this choice of $w$ is that here the parent of a block $B$ cannot be much larger than $B$ itself.
\begin{lem}
\label{lem:size_vs_parent}
Let $(G_0, G_1, \ldots, G_r)$ be an $(n, d, k, 2k)$-system of grids. Then for any $0 \leq i \leq r$ and any $G_i$-block $B$ it holds that $|\Par(B)| / |B| < 3^d$. 
\end{lem}
\begin{proof}
For $i = r$ this is trivial. Now fix $i < r$ and let $B$ be a $G_i$-block. Recall that, following \eqref{eqn:explicit_representation}, one can write $B = \prod_{j=1}^{d} I_{i_j}[k\colon ]$ where each interval $I_{i_j}$ (for $j \in [d]$) is of size at least $2k \geq 4$. On the other hand, we can also write $\Par(B) = \prod_{j=1}^{d} I'_{i'_j}[k\colon ]$ where $I'_{i'_j} \supseteq I_{i_j}$ for any $j \in [d]$. It is not hard to verify that $|I'_{i'_j}| \leq 2|I_{i_j}| + 1$ most hold, and so 
$$
\frac{|\Par{B}|}{|B|} = \prod_{j=1}^{d} \frac{|I'_{i'_j}| - (k-1)}{|I_{i_j}| - (k-1)} \leq \prod_{j=1}^{d} \frac{2|I_{i_j}| + 1 - (k-1)}{|I_{i_j}| - (k-1)} \leq \left(\frac{2 \cdot 2k - k+2}{2k - k + 1}\right)^d < 3^d
$$
where the second inequality holds since $|I_{i_j}| \geq 2k$ for any $j$.
\end{proof}

The next corollary follows immediately from Lemmas \ref{lem:only_parents} and \ref{lem:size_vs_parent}.
\begin{cor}
\label{cor:many_entries_in_maximal_witnesses}
Suppose that $A$ is $\epsilon$-far from $\mathcal{P}$. Then the total number of entries in the blocks of $\mathcal{W}$ is at least $\epsilon (n/3)^d$.
\end{cor}

We are now ready for the proof of the main upper bound of this paper, Theorem \ref{thm:strong_upper_bound}.
\subsection{Proof of Theorem \ref{thm:strong_upper_bound}}
\label{subsec:proof_upper_bound}
As before, we may assume that $k \leq \epsilon^{1/d} n / 10$. For larger $k$, the expression in the theorem dominates $n^{d}$ and thus becomes trivial.
Consider the $(n, d, k, 2k)$-system of grids $(G_0, G_1, \ldots, G_r)$ mentioned above. For any $0 \leq i \leq r$, define $\delta_i = |\mathcal{B}_i \cap \mathcal{W}| / |\mathcal{B}_i|$, where $\mathcal{B}_i$ is the set of all $G_i$-blocks. In other words, $\delta_i$ is the fraction of maximal witnesses among the $G_i$-blocks. 
By Corollary \ref{cor:many_entries_in_maximal_witnesses}, if $A$ is $\epsilon$-far from $\mathcal{P}$ then $\sum_{i=0}^{r} \delta_i \geq \epsilon / 3^d$. 
Define $r' = \lfloor \log(\epsilon^{1/d} n / k) \rfloor \geq 1$, noting that $G_{r'} \in \mathcal{G}(n, d, k, w_{r'})$ with $w_{r'} \geq 2k \cdot 2^{r'} \geq \epsilon^{1/d} n$. Thus, the total number of blocks in $\mathcal{B}_{r'}$ is bounded by $(n / w_{r'})^d \leq 1 / \epsilon$.

\paragraph{The test}
We iterate the following basic step $2 \cdot 3^d / \epsilon$ times.
\begin{enumerate}
	\item Pick $B \in \mathcal{B}_0$ uniformly at random and query all entries of $\overline{B}$. 
	\item For any $1 \leq i \leq r'$, pick $B \in \mathcal{B}_i$ uniformly at random and query all entries of $\bigcup_{B \in \mathcal{Q}_0} \partial{B}$. 
\end{enumerate}
Finally, the test rejects if and only if at least one of the blocks $B$ picked during the process is a $(\mathcal{P}, A)$-witness. (Recall that querying all boundary entries of a $G_i$-block for $i > 0$ suffices to determine whether it is unrepairable, and thus a witness.)

The test is clearly non-adaptive, and has one-sided error: It only rejects if it finds a witness. As we have seen earlier, all witnesses contain an $\mathcal{F}$-copy.
	
The test is \emph{canonical} in the following sense. The choice of queries in every basic step depends only on $n, d, k$, and (weakly) on $\epsilon$, and is independent of the property $\mathcal{P}$ or the alphabet $\Sigma$. To determine which entries constitute a block, it suffices to know the parameters of the block, that depend only on $n, d, k$; the dependence in $\epsilon$ is only taken into account in the choice of $r'$. The test only considers $\mathcal{P}$ in order to determine whether each queried block is a witness.

\paragraph{Analysis}
Suppose that $A$ is $\epsilon$-far from $\mathcal{P}$.
If $\delta_i > 0$ for some $i > r'$ then it must hold that $\delta_{r'} > 0$ as well (since any unrepairable $G_i$-block most contain an unrepairable $G_{i'}$-block for any $i' < i$). By the choice of $r'$, we must have $\delta_{r'} \geq 1 / |\mathcal{B}_{r'}| \geq \epsilon$ in this case. If the above doesn't hold, then $\delta_i = 0$ for any $i \geq r'$, implying that $\sum_{i=0}^{r'-1} \delta_i \geq \epsilon/3^d$. Therefore, in both cases, we have $\sum_{i=0}^{r'} \delta_i \geq \epsilon / 3^d$.

The probability that a random $\mathcal{B}_i$-block is a witness is at least $\delta_i$, and therefore the probability that a single basic step leads to a rejection of $A$ is at least $\sum_{i=0}^{r'} \delta_i \geq \epsilon / 3^d$. Running $2 \cdot 3^d / \epsilon$ independent iterations of the basic step ensures that the test will accept $A$ with probability at most $(1 - \epsilon / 3^d)^{2 \cdot 3^d / \epsilon} \leq e^{-2} < 2/3$, as desired.

\paragraph{Query complexity}
For $d=1$, the query complexity of each basic step is $O(k r')$: The test queries $\overline{B}$ for a single block $B \in \mathcal{B}_0$, and the boundaries of $r'$ larger blocks. Considering the parameters of our system of grids, we have $|B| \leq 4k$ and so $|\overline{B}| < 6k$. On the other hand, the boundary of each of the larger blocks is of size at most $2k-2$. Therefore, the total query complexity for the 1D test is $O(kr' / \epsilon) = O\left(\frac{k}{\epsilon} \log\left( \epsilon n / k \right)\right)$ as desired.

For $d > 1$, consider a single basic step, and for any $0 \leq i \leq r'$ let $B_i \in \mathcal{B}_i$ be the $\mathcal{G}_i$-block picked in this step.
From \eqref{eqn:block_size} we have $|\overline{B_0}| \leq (6k)^d$, while for any $i > 0$ we have $|\partial B_i| \leq 2d(k-1) (4k \cdot 2^i + k)^{d-1} = O(d \cdot (4k)^d \cdot 2^{(d-1)i})$. Note that the last expression grows exponentially with $(d-1)i$, so the total number of queries in a single basic step is $O((6k)^d + d \cdot (4k)^d 2^{(d-1)r'})$. Plugging in $r'$, we have $2^{(d-1)r'} = \frac{\epsilon^{(d-1)/d} }{k^{d-1}} n^{d-1}$. As the test runs $O(3^d /\epsilon)$ iterations of the basic step,
we conclude that the total query complexity is bounded by $c^d k \epsilon^{-1/d} n^{d-1}$ for an absolute constant $c > 0$, completing the proof of Theorem \ref{thm:strong_upper_bound}.

\subsection{Proximity oblivious test}
The proof of Theorem \ref{thm:proximity_oblivious} follows by a very simple modification of the proof of Theorem \ref{thm:strong_upper_bound}. The desired proximity oblivious test (POT) is the so called ``basic step'' from the above test, with $r$ replacing $r'$ (since $r'$ depends on $\epsilon$). The POT rejects if it infers that one of the blocks queried is a witness, like the above test. Its query complexity is $O(kr) = O(k \log{n/k})$ for $d=1$. In the case $d > 1$, the query complexity is dominated by the size of $\partial B_r$, which is bounded by $O(dkn^{d-1})$.

Clearly this POT has one-sided error, and its queries do not depend on the property $\mathcal{P}$ and the alphabet $\Sigma$ (on the other hand, they do depend on $n, d, k$). Using the notation of the previous subsection and denoting by $\epsilon_A$ the Hamming distance of a given input $A$ from $\mathcal{P}$, we get (exactly as in the beginning of Subsection \ref{subsec:proof_upper_bound}) a rejection probability of at least $\sum_{i=0}^{r} \delta_i \geq \epsilon_A / 3^d$ for $A$, which is linear in $\epsilon_A$ for fixed $d$. This concludes the proof.

\section{Lower bounds}
\label{sec:lower_bound}
In this section we prove Theorems \ref{thm:main_lower_bound} and \ref{thm:two_sided_lower_bound} for an alphabet of size $2^{O(n^d)}$. (In Section \ref{section:lower_shrinking} we show how to shrink the alphabet size for the proof of Theorem \ref{thm:main_lower_bound}.)
The first and main step is to prove a non-adaptive one-sided $\Omega(n^{d-1})$ lower bound and a non-adaptive two-sided $\Omega(n^{(d-1)/2})$ lower bound, both for the case of constant $\epsilon$ and $k=2$. Later on, we show how to achieve the correct dependence in $k$ and $\epsilon$ across the whole range of parameters.

\subsection{$k=2$ and constant $\epsilon$}
In this subsection we present a $2$-local property $\mathcal{P}$ of $[n]^d$-arrays that requires $\Omega(n^{d-1} / d)$ non-adaptive queries to test with one-sided error for constant $\epsilon > 0$ and $n \geq 6$. (To deal with smaller $n$ we may take $c \geq 5$ in the statement of the theorems). 
We start by providing a ``global'' description of $\mathcal{P}$, and only later show that it can actually be described as a $2$-local property.

The first coordinate of each entry $x \in [n]^d$ will be of special importance in the description of the property, and we designate a special name for it: The \emph{height} of the entry $x = (x_1, \ldots, x_d)$ is $x_1$. Accordingly, we define the \emph{floor} of $[n]^d$ as the set $\left\{(1, x_2, \ldots, x_d) : (x_2, \ldots, x_d) \in [n]^{d-1}\right\}$. The \emph{ceiling} is defined similarly as $\left\{(n, x_2, \ldots, x_d) : (x_2, \ldots, x_d) \in [n]^{d-1}\right\}$.

Very roughly speaking, arrays satisfying the property contain ``data'' in the floor and ceiling entries, and we show that the property of ``having the same data in the floor as in the ceiling'' can actually be expressed in a $2$-local way that is hard to capture using one-sided non-adaptive tests -- showing that $\Omega(n^{d-1} / d)$ queries are required to test the property for fixed $\epsilon$ and $k$. Relying on the birthday paradox, we show that even two-sided error non-adaptive tests require $n^{\Omega(d-1)} / d$ queries for this property.

The alphabet over which we work is $\Sigma = [n]^d \times [n]^d \times 2^{[2n^{d-1}]}$, where $2^X$ denotes the power set of a set $X$, i.e., the collection of all subsets of $X$. For an $[n]^d$-array $A$ over $\Sigma$, we interpret the value $A(x)$, for each entry $x \in [n]^d$, as a triple $(A_1(x), A_2(x), A_3(x))$, where $A_1(x), A_2(x) \in [n]^d$ are viewed as pointers to locations in the array and $A_3(x) \subseteq [2n^{d-1}]$ is viewed as a set.

\subsubsection{Global description of the property}
\label{subsection:global_description}
Formally, we say that $A$ satisfies $\mathcal{P}$ if and only if all of the following holds:
\begin{description}
	\item[Pointer to self] $A_1(x) = x$ for any $x \in [n]^d$.
	
	\item[Centers of gravity] There exists a special location $\ell = (\ell_1, \ldots, \ell_d) \in [n]^d$ which is called the \emph{lower center of gravity}, and satisfies $\ell_1 > 1$ and $\ell_1 + 1 < n$ (here we implicitly assumed that $n \geq 4$). Similarly, the tuple $u = (u_1, \ldots, u_d) \in [n]^d$ defined by $u_1 = \ell_1 + 1$ and $u_i = \ell_i$ for $2 \leq i \leq d$ is called the \emph{upper center of gravity}. 
	
	\item[Pointer to lower center of gravity] For any $x \in [n]^d$ we require that $A_2(x) = \ell$.
	
	\item[Data originates in floor and ceiling] 
	For any entry $x$ in the floor or the ceiling, the set $A_3(x)$ is a singleton, i.e., contains exactly one element from $[2n^{d-1}]$.
	
	\item[Data flows to center of gravity] Define a (simple) \emph{path} $(z^{(1)}, \ldots, z^{(t)})$ as a tuple of entries in $[n^d]$ such that each pair of entries $z^{(i)}$ and $z^{(i+1)}$ are neighbors (recall that $x = (x_1, \ldots, x_d), y = (y_1, \ldots, y_d) \in [n^d]$ are neighbors if $\sum_{j=1}^{d} |x_j - y_j| = 1$).
	
	For every entry $x = (x_1, \ldots, x_d) \in [n^d]$ not higher than $\ell$, let the \emph{flow} $\Gamma_{x}$ denote the unique path from $x$ to $\ell$ with structure as follows: The path starts from $z^{(0)} = x$. Now, for every $i$ from $1$ to $d$ (in this order, and including $1$ and $d$), the path continues strictly along the $i$-th coordinate until it reaches the entry $z^{(i)} = (\ell_1, \ldots, \ell_i, x_{i+1}, \ldots, x_d)$. Note that $z^{(i+1)}$ is indeed reachable from $z^{(i)}$ when moving only along the $i$-th coordinate, and that $z^{(d)} = \ell$. For $x \in [n]^d$ not lower than $u$, we define $\Gamma_x$ symmetrically, with all occurrences of $\ell$ and $\ell_i$ replaced by $u$ and  $u_i$ respectively. 
	
	Finally, for any $y \in [n]^d$ that is not a floor or ceiling entry, we require that $A_3(y) = \bigcup_{x \in [n]^d \colon \Gamma_{x} \ni y} A_3(x)$. In other words, the data is carried along flows, and the data kept in each entry is an aggregation of data from all flows passing through it.

	\item[Centers of gravity have same data]
	We require that $A_3(\ell) = A_3(u)$.
	
\end{description}

\subsubsection{Local description of the property}
\label{subsec:local_description}
We define a $2$-local property $\mathcal{P}'$ by describing the family $\mathcal{S}$ of all \emph{allowed} $2 \times \ldots \times 2$ consecutive subarrays. The corresponding forbidden family $\mathcal{F}$ is naturally defined as $\Sigma^{[2]^d} \setminus \mathcal{S}$. We later show that $\mathcal{P} = \mathcal{P}'$.

An array $S \colon [2]^d \to \Sigma$ is in $\mathcal{S}$ if and only if all of the following conditions hold. As before, we represent the value of $S$ on $x \in [2]^d$ as $(S_1(x), S_2(x), S_3(x))$, similarly to the roles that $A_1, A_2, A_3$ had earlier.

For what follows, define the \emph{difference} and the \emph{sum} between two tuples of integers $(x_1, \ldots, x_d)$ and $(y_1, \ldots, y_d)$ as the tuples $(x_1 - y_1, \ldots, x_d - y_d)$ and $(x_1 + y_1, \ldots, x_d + y_d)$ respectively. 
\begin{description}
	\item[Difference preservation] For any $x, y \in [2]^d$ we require $S_1(x) - S_1(y) = x - y$.  
	
	\item[Pointing to center of gravity] We require $S_2(x) = S_2(y)$ for any pair of neighbors $x, y \in [2]^d$. We also require that the first coordinate of $S_2(x)$ is between $2$ and $n-2$, for any $x$. 
	
	From this point onwards, define $\ell = S_2(x)$ where, say, $x = (1,1,\ldots,1)$, and set $u = \ell + (1, 0, \ldots, 0)$. 
	$\ell$ and $u$ will eventually be viewed as centers of gravity as above.

	\item[Data originates in floor and ceiling]
	For any $x$ where $S_1(x)$ is in the floor or the ceiling of $[n]^d$, we require that $S_3(x)$ is a singleton. 
	\item[Flow of data]  
	For any $x \in [2]^d$, let $\Pre(x)$ denote the set of all $y \in [2]^d$ with $S_1(x) \in \Gamma_{S_1(y)}$. 
	If $\{S_1(y) : y \in \Pre(x)\}$ contains all neighbors $z$ of $S_1(x)$ in $[n]^d$ that satisfy $x \in \Gamma_z$, then we require $S_3(x) = \cup_{y \in \Pre(x)} S_3(y)$.
	
	\item[Centers of gravity have same data] If $S_1(x) = S_2(x)$ and $S_1(y) = S_1(x) + (1, 0, \ldots, 0)$, then we require $S_3(x) = S_3(y)$.
	
\end{description}
It is straightforward to verify that an array satisfying $\mathcal{P}$ also satisfies $\mathcal{P}'$. The other direction is only slightly more challenging.

\begin{lem}
Any array $A \colon [n]^d \to \Sigma$ satisfying $\mathcal{P}'$ also satisfies $\mathcal{P}$.
\end{lem}
\begin{proof}
We generally use the same notation as in Subsection \ref{subsection:global_description}. 
First we show that $A_1(x) = x$ for any $x \in [n]^d$. By iterating the difference preservation rule of $\mathcal{P}'$, one can conclude that $A_1(x) - A_1(y) = x-y$ for any $x$ and $y$. In other words, $A_1$ is a difference preserving function from $[n]^d$ to itself, implying that it must be the identity. 

The second local condition trivially implies that $A_2(x) = A_2(y)$ for any $x, y \in [n]^d$. Thus, the lower and upper centers of gravity $\ell $ and $u$ are well defined. The third local condition implies that $A_3(x)$ is a singleton for any floor or ceiling entry $x$. The last local condition implies that we must have $A_3(\ell) = A_3(u)$. 

It remains to show that the local condition on data flows implies the required global behavior. We use the transitive nature of paths, that follows from the way they are constructed: If $z \in \Gamma_y$ and $y \in \Gamma_x$ then $z \in \Gamma_x$. Given any $z \in [n^d]$, there exists a $2 \times \ldots \times 2$ subarray of $A$ containing $z$ and all neighbors $y$ of $z$ for which $z \in \Gamma_y$. The local rule implies that $A_3(z) = \bigcup A_3(y)$ where the union is over all such $y$. By iteratively applying the local rule and using the transitivity, we get that $A_3(z) = \bigcup_{x \in [n]^d \colon \Gamma_x \ni z} A_3(x)$, as desired.
\end{proof}

\subsubsection{Proving the lower bounds}

Here we prove the following lower bounds regarding the testability of $\mathcal{P}$.
\begin{thm}
Fix $\epsilon \leq 1/4$. Any non-adaptive $\epsilon$-test for the $2$-local property $\mathcal{P}$ requires
\begin{enumerate}
\item $\Omega(n^{d-1}/d)$	queries if the test has one-sided error.
\item $\Omega(n^{(d-1)/2} / d)$ queries if the test has two-sided error.
\end{enumerate}
\end{thm}
%The constant hidden in the $\Omega_d$-term is exponential in $d$, and we do not try to optimize it.

As is standard in proving lower bounds in property testing, we will use Yao's minimax principle \cite{Yao77} to prove the theorem. 
Define $\mathcal{C}_{\rej}$ as the collection of all $[n]^d$-arrays $A$ over $\Sigma$ satisfying the following three conditions.
\begin{itemize}
	\item $A$ satisfies the first five conditions of $\mathcal{P}$ described in Subsection \ref{subsection:global_description}. Consequently, $A$ has well defined centers of gravity $\ell = (\ell_1, \ldots, \ell_d)$ and $u = (u_1, \ldots, u_d)$. 
	\item On the other hand, $A$ does not satisfy the last condition of Subsection \ref{subsection:global_description}. Instead we have $|A_3(\ell)| = |A_3(u)| = n^{d-1}$ and $A_3(\ell) \cap A_3(u) = \emptyset$. In particular, the singletons $A_3(x)$ where $x$ is a floor or ceiling entry are all pairwise disjoint.
	\item $n/3 < \ell_1 < u_1 \leq 2n/3$. (This requires $n \geq 6$.)
\end{itemize}
Clearly, $\mathcal{C}_{\rej}$ is not empty. We also define $\mathcal{C}_{\acc}$ as the set of all arrays $A$ satisfying $\mathcal{P}$ and $n/3 < \ell_1 < u_1 \leq 2n/3$, where $\ell = (\ell_1, \ldots, \ell_d)$ and $u = (u_1, \ldots, u_d)$ are the centers of gravity of $A$.
\begin{lem}
\label{lem:eps_far_distrib}
	All arrays in $\mathcal{C}_{\rej}$ are $1/4$-far from $\mathcal{P}$.
\end{lem}
\begin{proof}
	Let $A \in \mathcal{C}_{\rej}$ and let $A'$ be any array satisfying $\mathcal{P}$. We need to show that $A$ is $1/4$-far from $A'$.
	By definition, $A_1(x) = A'_1(x) = x$ for any $x \in [n]^d$. If $A_2(x) \neq A'_2(x)$ for some $x$, then the same holds for any $x$ (recall that both functions $A_2$ and $A'_2$ are constant on their domain $[n]^d$), implying that one has to modify all entries in $A$ to make it equal to $A'$. Therefore, from now on we may assume that $A$ and $A'$ have the same centers of gravity $\ell$ and $u$. Consequently, the structure of all flows $\Gamma_x$ in $A$ must be identical to the structure of the corresponding flows in $A'$. It remains to show that $A_3$ and $A'_3$ differ on at least $n^d / 4$ entries.
	
	The second condition on $A$ states that $|A_3(\ell)| = |A_3(u)| = n^{d-1}$ but $A_3(\ell) \cap A_3(u) = \emptyset$. On the other hand, $A'$ must satisfy $A'_3(\ell) = A'_3(u)$. Thus, to turn $A$ into $A'$, one must modify the values of at least $n^{d-1}$ floor and ceiling singletons. Given a floor entry $x = (1, x_2, \ldots, x_d) \in [n^d]$, modifying $A_3(x)$ forces us to also modify all entries of the form $A_3(x')$ for any $x' = (x_1, x_2, \ldots, x_d)$ with $x_1 \leq \ell$, since $A'_3(x) = A'_3(x')$ for any such $x'$. In total, the number of modifications needed to turn $A$ into $A'$ is at least $n^{d-1} \cdot \min\{\ell, n-u\} \geq n^d / 3$. This holds for any $A'$ satisfying $\mathcal{P}$, so $A$ is $\epsilon/4$-far from $\mathcal{P}$.
\end{proof}

For any subset $Q \subseteq [n]^d$ and any array $A$ satisfying the first five conditions of Subsection \ref{subsection:global_description}, we denote by $I(Q, A)$ the set of all floor and ceiling entries $x \in [n]^d$ for which the path $\Gamma_x$ intersects $Q$. (Recall that the behavior of $\Gamma_x$ depends only on the center of gravity.)

\begin{lem}
	\label{lem:small_expected_size}
For any subset $Q \subseteq [n]^d$, the expected size of $I(Q, A)$ when $A$ is picked uniformly at random from $\mathcal{C}_{\rej}$ is at most $7 d |Q|$. The same holds when $A$ is picked u.a.r.\@ from $\mathcal{C}_{\acc}$.
\end{lem}
\begin{proof}
Let $x = (x_1, \ldots, x_d) \in [n]^d$. If $x_1 \leq n/3$ or $x_1 > 2n/3$, then $x$ does not share the first coordinate with any $A$ in $\mathcal{C}_{\acc}$ or $\mathcal{C}_{\rej}$, and thus $A_3(x)$ is a singleton for any such $A$. Therefore, $x$ contributes at most $1$ to the expected size of $I(Q, A)$ in this case.

Suppose now that $n/3 < x_1 \leq 2n/3$. For $A$ chosen u.a.r. from either $\mathcal{C}_{\acc}$ or $\mathcal{C}_{\rej}$, and any integer $1 \leq i \leq d$, the probability that $x$ shares its first $i$ coordinates with the lower or upper centers of gravity of $A$ is bounded by $ \frac{2}{n/3} \cdot \frac{1}{n^{i-1}} = \frac{6}{n^i}$. If this event does not hold, then $|\Gamma_x| \leq n^{i-1}$. Therefore, the expected size of $\Gamma_x$ is bounded by $1 + \sum_{i=1}^{d} \frac{6}{n^{i-1}} \cdot n^{i-1} = 6d+1$. 
Since $|I(Q, A)| \leq \sum_{x: \Gamma_x \cap Q \neq \emptyset} |\Gamma_x|$, the claim follows.
\end{proof}

\begin{proof}[Proof of Theorem \ref{thm:main_lower_bound}]
The main argument of the proof is that any one-sided error test for $\mathcal{P}$ cannot reject an array $A \in \mathcal{C}_{\rej}$ if the set of queries $Q$ made by the test satisfies $|I(Q,A)| \leq n^{d-1}$. Indeed, observe that for any such $A$ and set $I(Q, A)$, there exists an array $A' \in \mathcal{C}_{\acc}$ satisfying $A'(x) = A(x)$ for any $x \in I(Q, A)$; $A'$ can be created from $A$ by modifying  values of some of the singletons $A_3(x)$ for floor and ceiling entries $x \notin I(Q, A)$ to make unions of singletons in the floor identical to that of the ceiling. Note that as a result we also need to update the $A'_3$-values along the relevant paths; by definition of $I(Q, A)$, these paths do not intersect $Q$ so this is not a problem. Consequently, the array $A'$ satisfies $A'(x) = A(x)$ for any $x \in Q$, so querying the entries of $Q$ does not suffice to distinguish between $A$ and $A'$.    

By Lemma \ref{lem:small_expected_size} and Markov's inequality, any set $Q \subseteq [n]^d$ of size at most $ n^{d-1} / 22 d$ satisfies $|I(Q, A)| \leq n^{d-1}$ with probability bigger than $1/3$. Thus, any non-adaptive one-sided error test for $\mathcal{P}$ must make $ \Omega(n^{d-1} / d)$ queries.
\end{proof}

\begin{proof}[Proof of Theorem \ref{thm:two_sided_lower_bound}]
Let $\mathcal{D}$ be a distribution of $[n]^d$-arrays over $\Sigma$, where with probability $1/2$ the array $A$ is taken u.a.r. from $\mathcal{C}_{\acc}$, and otherwise, $A$ is taken u.a.r. from $\mathcal{C}_{\rej}$. We will prove that any two-sided test $T$ whose success probability over the distribution $\mathcal{D}$ is at least $2/3$ must make at least $O(n^{(d-1)/2} / d)$ queries. By Lemma \ref{lem:eps_far_distrib}, this suffices to prove Theorem \ref{thm:two_sided_lower_bound}.

For two arrays $A, A' \in \mathcal{C}_{\rej} \cup \mathcal{C}_{\acc}$ with the same centers of gravity $\ell$ and $u$, we say that $A$ and $A'$ are \emph{equivalent} if there exists a permutation $\pi$ on $[2n^{d-1}]$ so that $A'_3(x) = \pi(A_3(x))$ for any $x \in [n]^d$, where we define $\pi(X) = \{ \pi(x) : x \in X\}$ for any set $X \subseteq [n]^d$.
Clearly, $A \in \mathcal{C}_{\rej}$ if and only if $A' \in \mathcal{C}_{\rej}$, and the same holds with respect to $\mathcal{C}_{\acc}$.  
 
By the above symmetry, it suffices to only consider tests $T$ whose behavior is invariant under equivalence (i.e., accepts $A$ if and only if it accepts $A'$). Indeed, one can turn any valid two-sided test $T$ for the distribution $\mathcal{D}$ into a valid two-sided error non adaptive equivalence-invariant test $T'$ with the same number of queries by doing the following: First $T'$ applies a random permutation $\pi$ on $[2n^{d-1}]$ to all elements of all $A_3$-sets it queries, and then it returns the same answer that $T$ would have returned on the array with the permuted values.   
	
An important observation is that an equivalence-invariant test $T$ for $\mathcal{D}$ must decide whether to accept or reject based only on the following parameters:
\begin{itemize}
	\item The location of the centers of gravity (note that $T$ knows $\ell$ and $u$ after making its queries, and that they completely determine how the $A_3$-data flows).
	\item The total number $N(Q, A)$ of \emph{collision pairs} $(x, y)$ satisfying $A_3(x) = A_3(y)$,  where $x \in I(Q, A)$ is a floor entry and $y \in I(Q, A)$ is a ceiling entry.
\end{itemize} 
	
For the next step we utilize a birthday-paradox type argument. For an array $A \in \mathcal{C}_{\rej}$, $N(Q, A) = 0$ always holds, regardless of the locations of the centers of gravity. For $A \in \mathcal{C}_{\acc}$, the probability that $N(Q, A) > 0$ is bounded by $|I(Q, A)|^2 / n^{d-1}$, again regardless of the location of the centers of gravity. Note that this bound still holds if we restrict ourselves to a subset of $\mathcal{C}_{\acc}$ containing only arrays $A$ with a pre-specified lower center $\ell$; we shall use this fact later on.
 
Note that the distributions of the centers' locations in $\mathcal{C}_{\acc}$ and in $\mathcal{C}_{\rej}$ are identical: In both cases, the location of the lower center of gravity $\ell = (\ell_1, \ldots, \ell_d)$ is uniform among all entries satisfying $n/3 < \ell_1 \leq 2n/3 - 1$. 

We shall now finally prove by contradiction that $\Omega(n^{(d-1)/2} / d)$ queries are required for any equivalence-invariant two-sided error test $T$ on the distribution $\mathcal{D}$. By Yao's principle, we may assume that $T$ is deterministic, that is, it always makes the same set of queries $Q \subseteq [n]^d$, where $|Q| < c n^{(d-1)/2} / d$ for a small enough constant $c > 0$. 
Define $\mathcal{D}_{\ell}$ as the restriction of the distribution $\mathcal{D}$ to arrays whose lower center of gravity is $\ell$. When $T$ receives an array $A$ whose lower center is $\ell$, it must decide whether to accept or reject based only on $\ell$ and the value of $N(Q, A)$. However, by the above discussion, the probability that $N(Q,A) = 0$ for an array chosen from $\mathcal{D}_{\ell}$ is at least $\frac{1}{2} + \frac{1}{2} \left(1 - |I(Q, A)|^2  /n^{d-1}\right)$. Clearly, if we choose to accept when $N(Q, A) = 0$, then the error of our test restricted to $\mathcal{D}_{\ell}$ is at least $1/2$, since all instances in $\mathcal{D}_{\ell}$ that should be rejected are incorrectly accepted. We shall next show that for most choices of $\ell$, choosing to reject when $N(Q, A) = 0$ will also result in a large error, stemming from the fact that many should-be-accepted instances are incorrectly rejected.

Indeed, using Lemma \ref{lem:small_expected_size} and applying Markov's inequality twice, we conclude that for at least a $9/10$-fraction of the possible values of $\ell$ (we call these the \emph{good} values of $\ell$), at least a $19/20$-fraction of the arrays $A \in \mathcal{C}_{\acc}$ with lower center $\ell$ satisfy $|I(Q, A)| < n^{(d-1)/2} / 10$. Among arrays $A$ satisfying the last inequality, only a fraction of $|I(Q, A)|^2 / n^{d-1} \leq 1/100$ of the should-be-accepted $A$'s in $\mathcal{D}_{\ell}$ satisfy $N(Q, A) > 0$. 
For any good $\ell$, taking a union bound implies that $N(Q, A) = 0$ with probability $9/10$ among the accepting instances in $\mathcal{D}_{\ell}$. Thus, the error of choosing to accept when $N(Q, A) = 0$ when $\ell$ is good is at least $\frac{1}{2} \cdot \frac{9}{10} = \frac{9}{20}$. Therefore, the total error of $T$ on the distribution $\mathcal{D}$ is at least $9/20 - 1/10 > 1/3$, where the $1/10$ term corresponds to those $\ell$ values that are not good. This is a contradiction, since a valid two-sided error test must have error at most $1/3$ on $\mathcal{D}$.
\end{proof}

\subsection{Large $k$}
\label{subsec:large_k}
Efficiently generalizing the above construction to general $k$ is not difficult, giving the following bounds. 

\begin{thm}
\label{thm:lower_bound_large_k_fixed_eps}
	Fix $\epsilon \leq 1/4$. Any non-adaptive $\epsilon$-test for the $k$-local property $\mathcal{P}_k$ requires
	\begin{enumerate}
		\item $\Omega(k n^{d-1}/d)$	queries if the test has one-sided error.
		\item $\Omega(\sqrt{k} n^{(d-1)/2} / d)$ queries if the test has two-sided error.
	\end{enumerate}
\end{thm}

We may assume that $k$ is even and $k \leq n/6$. Since the main ideas and the proofs are almost identical to the case $k=2$, we only briefly sketch the main differences. 

The \emph{$r$-floor} of $[n]^d$ is defined as the set of all entries $(x_1, \ldots, x_d)$ with $1 \leq x_1 \leq r$, and the $r$-ceiling consists of all entries with $n-r+1 \leq x_1 \leq n$. Here we work with $r = k/2$.
We take $\Sigma = [n]^d \times [n]^d \times 2^{[kn^{d-1}]}$.

Recall the property $\mathcal{P}$ formally defined in Subsection \ref{subsection:global_description}. We now describe the lower bound property $\mathcal{P}_k$, generalizing the definition of the original property $\mathcal{P} = \mathcal{P}_2$.
As before, we read the value of an $[n]^d$-array $A$ in an entry $x \in [n]^d$ as a tuple $(A_1(x), A_2(x), A_3(x))$. An array $A$ satisfies $\mathcal{P}_k$ if the following conditions hold.
\begin{description}
	\item [Pointer to self] $A_1(x) = x$ for any $x \in [n]^d$.
	\item [Centers of gravity] defined as in the case $k=2$ and denoted $\ell = (\ell_1, \ldots, \ell_d)$ and $(u_1, \ldots, u_d)$. We additionally require that $k/2 \leq \ell_1 \leq n-k/2$, so $u_1 \leq n-k/2+1$. 
	\item [Pointer to lower center of gravity] for any $x \in [n]^d$, $A_2(x) = \ell$.
	\item [Data originates in $k/2$-floor and $k/2$-ceiling] For any entry $x$ in the $k/2$-floor or $k/2$-ceiling, $A_3(x)$ contains exactly one element from $[kn^{d-1}]$.
	\item [Data flows in jumps to center of gravity] The data flows in ``jumps'' of height $k/2$ in the first coordinate (in the rest of the coordinates the data flow is exactly as before).  More formally, we define the path $\Gamma_x$ from any $x$ not higher than the lower center of gravity $\ell$ as follows: 
	\begin{itemize}
		\item We start by moving from $x = (x_1, \ldots, x_d)$ upwards (in the first coordinate) in jumps of $k/2$. For example, the entry coming after $x$ in the path is $(x_1+k/2, x_2, \ldots, x_d)$. This part of the path ends when reaching an entry $y = (y_1, x_2, \ldots, x_d)$ with $\ell_1 - k/2 < y_1 \leq \ell_1$.
		\item The rest of the path continues like in the case $k=2$: For any $i=2,3,\ldots,d$ in this order, we continue the path by appropriately adding or subtracting $1$ from the $i$-th coordinate until its value reaches $\ell_i$.
	\end{itemize}
	The definitions regarding paths above the upper center of gravity $u = (u_1, \ldots, u_d)$ are symmetric; jumps of $k/2$ in the first coordinate are made downwards.
	\item [Same data around centers of gravity] We require that the union of the sets $A_{x}$ for $x = (\ell_1-i, \ell_2, \ldots, \ell_d)$, $0 \leq i \leq k/2 - 1$ is identical to the union of the sets $A_x$ for $x = (u_1+i, u_2, \ldots, u_d)$ for $0 \leq i \leq k/2-1$. This generalizes the ``centers of gravity have same data'' condition from the case $k=2$.
	
\end{description}

It is straightforward to verify that $\mathcal{P}_k$ is a $k$-local property, as was done in Subsection \ref{subsec:local_description} for $k=2$. 
To prove the lower bounds, we take a distribution similar to $\mathcal{D}$, but here the $k/2$-floor keeps $n^{d-1}k/2$ different values in its $A_3$-sets, and so does the $k/2$-ceiling. We take $\mathcal{C}_{\acc}$ as a collection of arrays where the union of the data in the $k/2$-floor is identical to the corresponding union in the $k/2$-ceiling. For the $\epsilon$-far $\mathcal{C}_{\rej}$ collection we take these sets of data to be disjoint. Lemma \ref{lem:small_expected_size} remains true as is (maybe with a modification in the constant, but without any dependence on $k$ in the bound). The rest of the proof is completely analogous, except that in the proof of the one-sided case the $n^{d-1}$ bound on the size of $I(Q, A)$ should be replaced with $n^{d-1}k/2$; and in the two-sided error case, the probability that $N(Q, A) > 0$ is now bounded by $2|I(Q,A)|^2 / n^{d-1}k$ instead of $|I(Q, A)|^2 / n^{d-1}$, allowing us to increase the size of $I(Q,A)$ (and consequently, of $Q$) by a $\Theta(\sqrt{k})$ factor.

\subsection{Small $\epsilon$}
\label{subsec:small_eps}
Proving the correct dependence in $\epsilon$ in the lower bound is very simple given the previous steps. We may assume that, say, $k \leq  \epsilon^{1/d} n / 5$. We partition $[n]^d$ into disjoint $d$-dimensional consecutive blocks, each of dimensions $n_1 \times n_2 \times \ldots \times n_d$ where $|n_i - \epsilon^{1/d} n| < 1$ for any $i$. There are $\Theta(1 / \epsilon)$ such blocks: recall that $\epsilon^{1/d} n \geq d$, so each such block is of size at least $\epsilon n^d \cdot (1-1/d)^d = \Omega(\epsilon n^d)$ and the $O(\epsilon n^d)$ upper bound is obtained similarly. 

Our lower bound distribution is as follows. We pick one of the blocks uniformly at random, and embed the lower bound construction of the last subsection, with parameters $n', d, k$ where the width is $n' = \lfloor \epsilon^{1/d} n \rfloor \geq d$, in the block. Note that $ \lfloor n' \rfloor \geq (1-1/d) \epsilon^{1/d} n $. It is easy to verify that for the one-sided error case, we need to make an expected number of $\Omega(kn'^{d-1} / d) = \Omega(k \epsilon^{1-1/d} n^{d-1} / d)$ queries in every block to test the property with constant success probability (where the expectation is over blocks). Since there are $\Theta(1/\epsilon)$ blocks, the total query complexity is $\Omega(k \epsilon^{-1/d} n^{d-1} / d )$ as desired.   
For the two-sided case, we need to make an expected number of $\Omega(\sqrt{k} n'^{(d-1)/2} / d) = \Omega(\sqrt{k} \epsilon^{(d-1)/2d} n^{(d-1)/2} / d)$ queries per block, and $\Omega(\sqrt{k} \epsilon^{(-d-1)/2d} n^{(d-1)/2} / d)$ in total.

\section{Lower bounds with smaller alphabet}
\label{section:lower_shrinking}

In this section we show how to slightly modify the lower bound construction of Section \ref{sec:lower_bound} in order to prove Theorem \ref{thm:main_lower_bound} with alphabet size that is polynomial in $n^d$.
We start by describing how to construct an ``alphabet-efficient'' analogue of the property $\mathcal{P}$ from Subsection \ref{subsection:global_description}, settling the case $k=2$. Then, we utilize ideas from Subsection \ref{subsec:large_k} to obtain a suitable construction for larger (even) values of $k$.

\subsection{$k=2$ and constant $\epsilon$}
\label{subsec:alphabet_efficient_simple}
Recall the lower bound property $\mathcal{P}$ presented in Subsection \ref{subsection:global_description}. The alphabet of the field $A_3(x)$ which indicates the ``data'' present at each point $x \in [n]^d$ is of size $2^{2n^{d-1}}$ (for the generalized property $\mathcal{P}_k$, the size increases to $2^{kn^{d-1}}$). We now show how to replace the ``data'' alphabet with an alternative one of size $n^{O(d)}$ without affecting the rest of the one-sided proof. For the two-sided case, the new proof does not work, and it remains open whether one can design a property that is hard to test non-adaptively with two sided error over an alphabet of reasonable size.

The main idea is to turn the data into counts of zeros and ones; the rest of the construction (in particular, the structure of flows) remains as before. Each floor and ceiling entry either contains a single zero or a single one (that is, its zeros/ones count is either 1/0 or 0/1), and the zero/one-count of each internal entry is the aggregation of the counts of all neighboring entries flowing into it.

Finally, the ``centers of gravity have the same data'' condition is replaced with a global condition on the 0/1-counts that is hard to check with one-sided error. For example, the condition that the total count of zeros is larger than the total count of ones. 

Formally, we make the following changes to the property from Subsection \ref{subsection:global_description}.
\begin{itemize}
\item The alphabet for $A_3$ is replaced with $\{0,1,\ldots, n^{d-1}\}^2$. For any $x$, $A_3(x)$ is viewed as the tuple containing the count of zeros $A^0_3(x)$ and ones $A^1_3(x)$, respectively, aggregated at $x$.
\item {\bf Data in floor and ceiling}: The condition ``data originates in floor and ceiling'' is replaced with the condition that $A_3(x) = (1,0)$ or $A_3(x) = (0, 1)$ for any floor and ceiling element. In other words, any floor or ceiling entry either holds a single occurrence of zero and no occurrences of one, or vice versa.
\item {\bf Data flow}: The condition regarding flow of data to the center of gravity remains like in $\mathcal{P}$, but we replace the requirement that $A_3(y) = \bigcup_{x \in [n]^d \colon \Gamma_{x} \ni y} A_3(x)$ in the original construction with the conditions $A^0_3(y) = \sum_{\Gamma_x \ni y} A^0_3(x)$ and $A^1_3(y) = \sum_{\Gamma_x \ni y} A^1_3(x)$. In other words, $A_3(y)$ aggregates the counts from $A_3(x)$ for all values of $x$ flowing to $y$.
\item {\bf More zeros than ones}: The condition ``centers of gravity have same data'' is replaced with the condition that the total count of zeros in the centers of gravity is larger than the total count of ones in them. 
\end{itemize}

Note that the property is still $2$-local (the modifications we made do not affect the locality). 
The collection $\mathcal{C}_{\text{rej}}$ we take consists of all arrays that satisfy all conditions except for the ``more zeros and ones'' condition, that is replaced by the requirement that the total count of zeros in the centers of gravity is exactly $\lfloor n^{d-1} / 4 \rfloor$, and additionally, $n/3 < \ell_1 < u_1 \leq 2n/3$. Clearly, all arrays in $\mathcal{C}_{\text{rej}}$ are $\Omega(1)$-far from satisfying the property.
The collection $\mathcal{C}_{\text{acc}}$ is defined similarly, but we replace the required count of zeros to $\lceil 3n^{d-1} / 4 \rceil$. Clearly, all arrays in the collection satisfy the property.

The rest of the proof follows in a straightforward manner.
In Lemma \ref{lem:small_expected_size}, we replace the quantity ``expected size of $I(Q, A)$'' with the quantity $C(Q, A) = \sum_{x \in Q} A^0_3(x) + A^1_3(x)$. The exact same proof gives that the expected value of $C(Q, A)$ is bounded by $7 d|Q| $ for any $Q \subseteq [n]^d$. Since we must have $C(Q, A) \geq \lfloor n^{d-1} / 2 \rfloor$ in order to distinguish (with one sided error) between $\mathcal{C}_{\text{acc}}$ and $\mathcal{C}_{\text{rej}}$, it follows (as in the original proof) that we must have $|Q| = \Omega(n^{d-1} / d)$, concluding the proof.

\subsection{General $k$ and $\epsilon$}

We now describe the ``alphabet-efficient'' lower bound construction for a general even $k \geq 2$ and constant $\epsilon$. The case of small $\epsilon$ is handled exactly as in Subsection \ref{subsec:small_eps}.

First, the alphabet of $A_3$ is taken as $\{0,1,\ldots,n^{d-1}\}^2$ (similarly to Subsection \ref{subsec:alphabet_efficient_simple}), where for any $x$, $A_3(x) = (A^0_3(x), A^1_3(x))$ is viewed as the aggregated count of zeros and ones, respectively, at $x$.

The first three conditions from Subsection \ref{subsec:large_k} are taken as is. The fourth condition is replaced by the requirement that for any $x$ in the $k/2$-floor or $k/2$-ceiling, either $A_3(x) = (1,0)$ which means that $x$ holds a single occurrence of zero, or $A_3(x) = (0,1)$, meaning that $A_3(x)$ holds a single occurrence of one.

For the fifth condition, the structural behavior of the data flow is exactly as in Subsection \ref{subsec:large_k}: data flows in ``jumps'' of $k/2$ along the first coordinate, and in single steps along the other coordinates. For each $y \in [n]^d$ we require, as in Subsection \ref{subsec:alphabet_efficient_simple}, that $A_3(y)$ aggregates the $0/1$-counts from $A_3(x)$ for all entries $x$ flowing into $y$.

The condition of ``same data around centers of gravity'' from Subsection \ref{subsec:large_k} is replaced with a suitable ``more zeros than ones'' condition: we require that $\sum_{x} A^0_3(x) > \sum_{x} A^1_3(x)$ where $x$ ranges over the set $\{(\ell_1-i, \ell_2, \ldots, \ell_d) : 0 \leq i \leq k/2 - 1\} \cup \{(u_1+i, u_2, \ldots, u_d): 0 \leq i \leq k/2 - 1\}$.
 
As a natural extension of Subsections \ref{subsec:large_k} and \ref{subsec:alphabet_efficient_simple}, it is straightforward to verify that the property described here is $k$-local, and that the rest of the proof (for the one-sided error setting) follows essentially as above. 
 
 %%%%%%%%%%%%%%%%%%%%%%%%%%%%%%%%%%%%%%%%%%%%5

\section*{Acknowledgments}
The author would like to thank Frederik Benzing, Eric Blais, Eldar Fischer, Sofya Raskhodnikova, Daniel Reichman and C. Seshadhri for stimulating discussions, and the anonymous reviewers for helpful suggestions.

\end{document}